\documentclass[journal]{IEEEtran}

\usepackage{hyperref}
\hypersetup{hidelinks}
\usepackage{xcolor,soul,framed} %,caption

\colorlet{shadecolor}{yellow}
\usepackage[pdftex]{graphicx}
\DeclareGraphicsExtensions{.pdf,.jpeg,.png}

\usepackage{amsfonts,amssymb}
\usepackage{mathrsfs}
\usepackage{amsthm}
\usepackage[cmex10]{amsmath}
\usepackage{mathtools}
\usepackage{bbm}
\DeclareMathAlphabet{\mathbbb}{U}{bbold}{m}{n}
\usepackage{array}
\usepackage{physics}
\usepackage{cite}
\usepackage{makecell}
\usepackage{dashbox}
\usepackage{booktabs}
\usepackage{threeparttable}
\usepackage{colortbl}
\usepackage{comment}
\usepackage{flushend}
\usepackage{anyfontsize}

\theoremstyle{definition}

\newtheorem{theorem}{Theorem}
\newtheorem{lemma}[theorem]{Lemma}
\newtheorem{proposition}[theorem]{Proposition}
\newtheorem{corollary}{Corollary}[theorem]
\theoremstyle{remark}
\newtheorem{example}{Example}
\newtheorem{remark}{Remark}

\newcommand{\hermconj}  {^{\mathsf{H}}}
\newcommand{\trans}     {^{\mathsf{T}}}
\newcommand{\pha}[1]    {\underline{#1}}
\newcommand{\phaconj}[1]{\overline{\underline{#1}}}
\newcommand{\vect}[1]   {\boldsymbol{#1}}

\newcommand{\mat}[1]    {\boldsymbol{#1}}

\definecolor{mygray}{gray}{.9}
\newcommand{\tabitem}{~~\llap{\textbullet}~~}

\begin{document}
\bstctlcite{IEEEexample:BSTcontrol}
    %\title{Transient Stability and Instability of Multivariable Grid-Forming Droop Control}
    %\title{Multivariable Grid-Forming Control and \\Provable Transient Stability}
    %\title{Complex Droop Control: Multivariable Grid-Forming and Provable Transient Stability}
    %\title{Phase-Amplitude Multivariable Stability of Grid-Forming Converters}
    %\title{Complex Droop Control in Grid-Connected Converters: Transient Stability and Instability}
    %\title{Complex Droop Control and Transient Stability of Grid-Forming Converters}    
    %\title{Transient Stability of Grid-Forming Converters With Complex Droop Control}    
    %\title{Guaranteed Transient Stability of Grid-Forming Converters With Complex Droop Control}
    %\title{Guaranteed Transient Stability for Grid-Forming Converters With Complex Droop Control}
    \title{Quantitative Stability Conditions for Grid-Forming Converters With Complex Droop Control}    
    \author{Xiuqiang~He,~\IEEEmembership{Member,~IEEE,}
          Linbin~Huang,~\IEEEmembership{Member,~IEEE,}
          Irina Subotić,\\
          Verena~Häberle,~\IEEEmembership{Graduate Student Member,~IEEE,}
        and~Florian~Dörfler,~\IEEEmembership{Senior Member,~IEEE}% <-this % stops a space
  \thanks{This work was supported by the European Union’s Horizon 2020 research and innovation program under Grant 883985.}% <-this % stops a space
  \thanks{The authors are with the Automatic Control Laboratory, ETH Zürich, 8092 Zürich, Switzerland. Email:\{xiuqhe, linhuang, subotici, verenhae, dorfler\}@ethz.ch.}}

% ====================================================================
\maketitle

%\pagenumbering{gobble} % no page number

% === ABSTRACT ====================================================================
\begin{abstract}
In this paper, we analytically study the transient stability of grid-connected converters with grid-forming complex droop control, also known as dispatchable virtual oscillator control. We prove theoretically that complex droop control, as a state-of-the-art grid-forming control, always possesses steady-state equilibria whereas classical droop control does not. We provide quantitative conditions for complex droop control maintaining \textit{transient stability} (global asymptotic stability) under grid disturbances, which is beyond the well-established local (non-global) stability for classical droop control. For the \textit{transient instability} of complex droop control, we reveal that the unstable trajectories are bounded, manifesting as limit cycle oscillations. Moreover, we extend our stability results from \textit{second-order} grid-forming control dynamics to \textit{full-order} system dynamics that additionally encompass both circuit electromagnetic transients and inner-loop dynamics. Our theoretical results contribute an insightful understanding of the transient stability and instability of complex droop control and offer practical guidelines for parameter tuning and stability guarantees.
\end{abstract}

% === KEYWORDS ====================================================================
\begin{IEEEkeywords}
Complex droop control, complex frequency, dispatchable virtual oscillator control (dVOC), grid-forming control, transient stability.
\end{IEEEkeywords}

% === I. INTRODUCTION =============================================================
\section{Introduction}

\IEEEPARstart{T}{HE} ubiquitous penetration of converter-interfaced generation in power grids has been displacing a substantial number of synchronous generators. Grid-forming (GFM) responsibilities initially shouldered by synchronous generators for establishing and regulating grid frequency and voltage are increasingly being transferred to converter-interfaced generation systems. Therefore, grid-forming converters are becoming indispensable assets in modern power systems \cite{rossano2022gridforming}. Under large disturbances, GFM converters need to maintain transient stability and provide continuous services for power grid operation \cite{rosso2021gridforming}. Since transient stability involves a wide range of changes in the system operating state, the nonlinear dynamic characteristics of GFM converters play a significant role. Particularly, in comparison to small-signal local stability, transient stability faces greater challenges in nonlinear analysis and control design \cite{dorfler2023control}.

Concerning early-stage GFM control strategies, e.g., classical p/f and q/v droop control \cite{simpson2013synchronization,dorfler2012synchronization,ainsworth2013structure,schiffer2014conditions,wu2019design,huang2019transient} and virtual synchronous machine (VSM) \cite{shuai2019transient,schiffer2019global,choopani2020newmulti,rokrok2022transient}, their stability in islanded or grid-connected scenarios has been studied extensively, as summarized in Table~\ref{tab:literature-review}(a). In particular, since the development of classical droop control is based on the single-input single-output (SISO) decoupling and linear approximation of the network power flow, it is generally observed that classical droop control schemes exhibit good small-signal/local stability performance \cite{du2020comparative}, and their local stability can be guaranteed theoretically \cite{simpson2013synchronization,dorfler2012synchronization,ainsworth2013structure}. However, the large-signal/transient response behavior of classical droop control is usually inferior, with rigorous guarantees of large-signal/transient/global stability established only under assumptions of inductive networks and/or fixed voltage magnitudes \cite{schiffer2014conditions,schiffer2019global}. The VSM strategy, which follows the same droop philosophy in its building block, further introduces a synthetic inertia response. The resulting second-order control structure, however, may deteriorate transient stability \cite{pan2020transient}, where the region of attraction (ROA) and the critical fault clearing time (CCT) are typically the main concerns \cite{shuai2019transient}.

\begin{table}
\scriptsize
\centering
\caption{Past Studies on Stability of Typical GFM Controls}
\setlength{\tabcolsep}{4pt} % Default value: 6pt
\renewcommand{\arraystretch}{1.0} % Default value: 1
\begin{threeparttable}
\begin{tabular}{lllllll}
\hline\hline
Ref. & Year & Type & Assumption & Method & Result & Publ.\tnote{1} \\
\hline
\multicolumn{7}{c}{(a) Elementary GFM controls} \\
\hline
\makecell[l]{\cite{simpson2013synchronization}}
& 2012 
& \makecell[l]{1st-order \\ Kuramoto \\ model} 
& \makecell[l]{Volt. fixed, \\ network \\ inductive} 
& \makecell[l]{Linearization} 
& Local
& Ctrl. \\
\arrayrulecolor{black!10}\hline
\makecell[l]{\cite{dorfler2012synchronization}}
& 2012 
& \makecell[l]{2nd-order \\ Kuramoto \\ model} 
& \makecell[l]{Volt. fixed} 
& \makecell[l]{Singular \\ perturbation, \\ contraction} 
& Local
& Ctrl. \\
\hline
\makecell[l]{\cite{ainsworth2013structure}}
& 2013 
& \makecell[l]{p/f droop in \\ structure-\\preserving \\ network} 
& \makecell[l]{Volt. fixed, \\ network \\ inductive} 
& \makecell[l]{Graph-\\theoretic \\ methods, \\ linearization} 
& Local
& PS \\
\hline
\makecell[l]{\cite{schiffer2014conditions}}
& 2014 
& \makecell[l]{p/f and \\ q/v droop} 
& \makecell[l]{Network \\ inductive} 
& \makecell[l]{Lyapunov-\\like analysis} 
& \makecell[l]{Local and \\ global} 
& Ctrl. \\
\hline
\makecell[l]{\cite{wu2019design}}
& 2019 
& \makecell[l]{p/f droop} 
& \makecell[l]{Volt. fixed, \\ network \\ inductive} 
& \makecell[l]{Phase \\ portrait} 
& Global
& PE \\
\hline
\makecell[l]{\cite{huang2019transient}}
& 2019 
& \makecell[l]{Current-\\limited \\ p/f droop} 
& \makecell[l]{Volt. fixed, \\ network \\ inductive} 
& \makecell[l]{Power-angle \\ curve analysis} 
& ROA
& PS \\
\hline
\makecell[l]{\cite{shuai2019transient}}
& 2019 
& \makecell[l]{VSM} 
& \makecell[l]{Static q/v \\ relationship} 
& \makecell[l]{Energy \\ function} 
& ROA
& PS \\
\hline
\makecell[l]{\cite{schiffer2019global}}
& 2019 
& \makecell[l]{VSM}
& \makecell[l]{Volt. fixed, \\ network \\ inductive} 
& \makecell[l]{Multivariable \\ cell structure, \\ Leonov func.} 
& \makecell[l]{Almost \\ global}
& Ctrl. \\
\hline
\makecell[l]{\cite{choopani2020newmulti}}
& 2020
& \makecell[l]{VSM} 
& \makecell[l]{Volt. fixed, \\ network \\ inductive} 
& \makecell[l]{Energy \\ function} 
& ROA
& PS \\
\hline
\makecell[l]{\cite{rokrok2022transient}}
& 2022 
& \makecell[l]{Current-\\ limited \\ VSM} 
& \makecell[l]{Volt. fixed, \\ network \\ inductive} 
& \makecell[l]{Power-angle \\ curve analysis} 
& ROA
& PS \\
\arrayrulecolor{black}\hline
\multicolumn{7}{c}{(b) Advanced GFM controls} \\
\hline
\makecell[l]{\cite{colombino2019global} \\ \cite{gross2019effect} \\ \cite{subotic2021lyapunov} \\ \cite{he2023nonlinear}}
& \makecell{2019\\2019\\2021\\2022} 
& \makecell[l]{dVOC}
& \makecell[l]{Consistent \\ setpoints} 
& \makecell[l]{Lyapunov-\\like analysis} 
& \makecell[l]{Almost \\ global} 
& Ctrl. \\
\arrayrulecolor{black!10}\hline
\makecell[l]{\cite{yu2021comparative}}
& 2021 
& \makecell[l]{dVOC} 
& \makecell[l]{Network \\ inductive} 
& \makecell[l]{2D phase \\ portrait} 
& Global 
& PE \\
\hline
\makecell[l]{\cite{awal2021transient}}
& 2021 
& \makecell[l]{dVOC} 
& \makecell[l]{Network \\ static} 
& \makecell[l]{3D phase \\ portrait} 
& Global 
& PE \\
\hline
\makecell[l]{\cite{arghir2020electronic}}
& 2018
& \makecell[l]{Matching \\ control} 
& \makecell[l]{Grid signal \\ available} 
& \makecell[l]{Hamiltonian \\ analysis} 
& Global 
& PE \\
\hline
\makecell[l]{\cite{tayyebi2022gridforming}}
& 2022 
& \makecell[l]{Hybrid angle\\control} 
& \makecell[l]{Voltage \\ fixed} 
& \makecell[l]{Lyapunov \\ analysis} 
& \makecell[l]{Almost \\ global} 
& Ctrl. \\
\hline
\makecell[l]{\cite{subotic2022dualport}}
& 2022 
& \makecell[l]{Power-bal. \\ dual-port \\ GFM control} 
& \makecell[l]{Volt. fixed, \\ network \\ inductive} 
& \makecell[l]{LaSalle \\ invariance \\ principle} 
& Local 
& Ctrl. \\
\arrayrulecolor{black}
\hline\hline
\end{tabular}
\label{tab:literature-review}
 \begin{tablenotes}
        \scriptsize
        \item[1] Publication communities include control community (Ctrl.), power electronics community (PE), and power systems community (PS).
 \end{tablenotes}
\end{threeparttable}
\end{table}

In recent years, advanced GFM controls have been widely explored, e.g., dispatchable virtual oscillator control (dVOC) \cite{colombino2019global,gross2019effect,subotic2021lyapunov,he2023nonlinear,yu2021comparative,awal2021transient},  matching control \cite{arghir2020electronic}, hybrid angle control \cite{tayyebi2022gridforming}, and dual-port GFM control \cite{subotic2022dualport}, see Table~\ref{tab:literature-review}(b). These controls are well-designed to address various control objectives and ensure stability, particularly global stability. Among these controls, dVOC, developed in \cite{colombino2019global}, has attracted increasing attention due to its unique capabilities in synchronization and voltage stabilization. Our recent work \cite{he2022complex} has reformulated dVOC equivalently as complex-power complex-frequency droop control (i.e., \textit{complex droop control}), based on a novel concept complex frequency \cite{milano2022complex}. Complex frequency is a unifying/compact representation of angular frequency and the rate of change of voltage. In this sense, complex droop control links active and reactive power inputs to frequency and rate-of-change-of-voltage outputs with a non-trivial coupling. As a consequence, the control is structurally multivariable and also intrinsically nonlinear, in contrast to SISO linear loops in classical droop control or VSM. In particular, complex droop control can effectively handle the inherent coupling and nonlinearity in the network's active and reactive power flows. In \cite{colombino2019global,gross2019effect,subotic2021lyapunov,he2023nonlinear}, a more general global stability guarantee has been established for dVOC, which goes beyond the local stability of classical droop control and beyond the scenarios of inductive networks and/or fixed voltages. The stability results have been experimentally validated and often reproduced \cite{seo2019dispatchable,lu2022virtual}.

While the stability of dVOC (i.e., complex droop control) in converter-based \textit{islanded} systems has been well studied in \cite{colombino2019global,gross2019effect,subotic2021lyapunov,he2023nonlinear}, the stability in \textit{grid connection} remains unaddressed theoretically. In this regard, the grid voltage appears as a forced input to the nonlinear system dynamics, altering both the dynamic response and the steady-state solution, presenting different challenges in nonlinear stability analysis. As can be seen in Table~\ref{tab:literature-review}(b), the phase-portrait method has been applied in \cite{yu2021comparative,awal2021transient} to graphically observe the transient behavior of dVOC under grid disturbances. However, this method relies on numerically solving the transient trajectories, and it cannot yield analytical results. To the best of our knowledge, the analytical results of transient stability for dVOC in a grid-connected context remain absent in the literature. The absence of theoretical results hinders the application of dVOC in the control of grid-connected converters. Likewise, there is limited attention on the \textit{transient instability} behavior analysis for dVOC. Another big research gap exists when it comes to the analytical exploration of \textit{transient stability of full-order dynamics} in GFM converters in grid connection. Namely, many prior studies as in Table~\ref{tab:literature-review} have focused solely on the dynamics of GFM controls while overlooking both circuit electromagnetic transients and inner-loop dynamics. Established small-signal stability conditions that consider full-order dynamics, e.g., \cite{huanghinf2020,saidistributed2020,deypassivity2023}, cannot ensure transient stability due to their dependency on equilibrium points and limitation to the small-signal regime. To summarize, we still lack quantitative conditions for transient stability guarantees.

In this paper, our objectives are threefold: firstly, to explore the transient stability condition of complex droop control (i.e., dVOC) in a reduced-order grid-connected converter system; secondly, to investigate how the transient instability of complex droop control behaves; and thirdly, to examine the transient stability condition for the full-order system dynamics. The contributions of this study are summarized as follows:
\begin{itemize}
    \item We prove that complex droop control always possesses steady-state equilibrium points whereas classical droop control and VSM do not. Additionally, complex droop control maintains transient stability under our provided quantitative parametric stability conditions.
    \item We prove that the trajectories of complex droop control are bounded and the transient instability of complex droop control manifests as limit cycle oscillations.
    \item We prove that the full-order system dynamics, which encompasses GFM control dynamics, circuit electromagnetic dynamics, and inner-loop dynamics, maintain transient stability within a wide region of attraction, and we provide quantitative parametric conditions for the transient stability of the full-order system dynamics.
\end{itemize}

The remainder of this paper proceeds as follows. The description of complex droop control and the comparison with classical droop control are presented in Section II. Full-order and reduced-order system models are formulated in Section III, where the stability problem statement is also included. Section IV presents the transient stability and instability analysis for the reduced-order system, and Section V extends the stability results to the full-order system. Case studies are presented in Section VI, and Section VII concludes this paper.

\section{Preliminaries}

% We introduce complex droop control as a multivariable GFM control, in comparison to classical droop control, and demonstrate its equivalence to dVOC. To do so, we first recall some definitions of relevant complex variables/coordinates.

\subsection{Definitions of Relevant Complex Variables}

\textit{1) Voltage Vector and Complex Voltage:}
Under three-phase balanced conditions, the voltage at the converter terminal node can be expressed in $\alpha\beta$ rectangular coordinates as
\begin{equation}
\label{eq:voltage-vector-expression}
    \vect{v} \coloneqq [ v_{\alpha}\ \ v_{\beta} ]\trans \coloneqq [ v\cos\theta \ \ v\sin\theta ]\trans,
\end{equation}
where $v$ denotes the voltage amplitude and $\theta$ denotes the rotational phase-angle.
We define the complex voltage vector $\pha{v}$ in the associated complex vector coordinate as
\begin{equation}
\label{eq:complex-voltage}
    \pha{v} \coloneqq v_\alpha + jv_\beta = v \left(\cos\theta + j\sin\theta\right) = {v}{e^{j{\theta}}},
\end{equation}
where the underline indicates complex-valued variables.

\textit{2) Complex Angle:}
Consider the complex voltage vector in \eqref{eq:complex-voltage} with $v > 0$. We define the complex angle vector $\pha{\vartheta}$ \cite{milano2022complex} as
\begin{equation}
\label{eq:complex-angle}
    \pha{\vartheta} \coloneqq {u} + j{\theta} \coloneqq \ln{v} + j{\theta}, \quad \Rightarrow \quad \pha{v} = e^{\pha{\vartheta}},
\end{equation}
where $u \coloneqq \ln{v}$ is defined as voltage logarithm. The complex-angle $\pha{\vartheta}$ contains two-dimensional information on the amplitude and angle of the voltage\footnote{Concerning the physical unit of $\pha{\vartheta}$, we consider $v$ in per-unit and $\theta$ in radian so that $\pha{\vartheta}$ is dimensionless or with a unit similar to per-unit/radian.}. From the definition of the complex angle in \eqref{eq:complex-angle}, it follows that $\pha{v} = e^{\pha{\vartheta}}$, which serves as a transformation between complex-voltage coordinates and complex-angle coordinates.

\textit{3) Complex Frequency:}
Based on the time derivative of ${\pha{\vartheta}}$, we define the associated complex frequency $\pha{\varpi}$ \cite{milano2022complex} as
\begin{equation}
\label{eq:complex-frequency}
    \pha{\varpi} \coloneqq \dot {\pha{\vartheta}} = {\dot v}/v + j{\dot \theta} = {\dot{\pha v}}/\pha v =: {\varepsilon} + j{\omega},
\end{equation}
where $ \varepsilon = \dot{v}/{v}$ is the rate of change of voltage, and $\omega$ denotes the angular frequency. More precisely, $\varepsilon$ denotes the radial frequency that represents the change in the amplitude while $\omega$ represents the standard angular frequency. This outlines a geometrical interpretation of the complex frequency.

\textit{4) Normalized Complex Power:}
In conventional power systems, the balance of active power indicates frequency synchronization. To indicate ``complex-frequency" synchronization, we define a ``complex version" of power. As in \cite{he2022complex}, the complex, conjugated, and normalized power $\phaconj{\varsigma}$ is defined as
\begin{equation}
\label{eq:complex-power}
    \phaconj{\varsigma} \coloneqq \phaconj{s}/v^2 = (p  - jq)/v^2 = {\phaconj{v}}\,{\pha{i}}/v^2 = {\pha{i}}/{\pha{v}},
\end{equation}
where the overline indicates conjugate variables, $\pha{s}$ is standard complex power with $p$ and $q$ being active and reactive power, and $\pha{i}$ represents output current. For brevity of notation, we further define normalized active and reactive power as $\rho \coloneqq p/v^2 = \Re{\phaconj{\varsigma}},\, \sigma \coloneqq q/v^2 = -\Im{\phaconj{\varsigma}}$, which will be employed in complex droop control in the following.

\begin{table*}
\centering
\caption{Comparison Between Classical Droop Control and Complex Droop Control (for Dominantly Inductive Grids).}
\begin{tabular}{ccc}
\hline\hline
 & Classical droop control & Complex droop control (dVOC) \\
\hline
Controller &
$\begin{aligned}
    \dot v &= \eta \left(q^{\star} - q \right) + \eta \alpha \left(v^{\star} - v \right), \\
    \dot{{\theta}} &= {\omega}_0 + \eta \left(p^{\star} - p\right).
\end{aligned}$&
$\begin{aligned}
    \frac{\dot v}{v} &= \eta \left(\frac{q^{\star}}{v^{\star 2}} - \frac{q}{v^2} \right) + \eta \alpha \frac{v^{\star} - v}{v^{\star}}, \text{or} + \eta \alpha \frac{v^{\star 2} - v^2}{v^{\star 2}},\\
    \dot \theta &= \omega_0  + \eta \left(\frac{p^{\star}}{v^{\star 2}} - \frac{p}{v^2} \right).
\end{aligned}$ \\
\arrayrulecolor{black!10}\hline
Assumption & Operation around the nominal point. & None. \\
\hline
Design philosophy & \makecell[c]{Heuristic from synchronous machines or linear \\ proportional control concerning linearized power flow.} & \makecell[c]{Inspiration from consensus synchronization (for dVOC); \\ Complex-frequency reformulation as complex droop control.} \\
\hline
Differences &
\makecell[l]{
\tabitem Standard $v$ and $\theta$ polar coordinates; \\
\tabitem Standard $p$ and $q$ power feedback. } &
\makecell[l]{
\tabitem Complex-angle $(\ln v + j\theta)$ coordinate; \\
\tabitem Normalized $p/v^2$ and $q/v^2$ power feedback. } \\
\hline
Stability & \makecell[c]{
Small-signal stability is guaranteed but transient stability \\ not guaranteed (operation far away from the nominal point).} & \makecell[c]{
Both small-signal stability and transient stability \\ are theoretically guaranteed.} \\
\makecell[c]{Power sharing} & Yes. & Yes. \\
\arrayrulecolor{black}\hline \hline
\end{tabular}
\label{tab:comparison-droop-control}
\end{table*}

\subsection{From Classical Droop Control to Complex Droop Control}
\label{sec:complex-droop-control}

\textit{1) Classical Droop Control:}
We review classical droop control, which is seen as an elementary GFM control\footnote{There exist many variants of droop control, e.g., static q/v relationship \cite{pan2020transient}, first-order low-pass filter in active and reactive power feedback \cite{schiffer2014conditions}, or even VSM, which, however, share the same steady-state equilibrium.},
\begin{subequations}
\label{eq:classical-droop}
    \begin{align}
    \label{eq:classical-qv-droop}
        \dot v &= \eta \left(q_{\varphi}^{\star} - q_{\varphi} \right) + \eta \alpha \left(v^{\star} - v \right), \\
    \label{eq:classical-pf-droop}
        \dot{{\theta}} &= {\omega}_0 + \eta \left(p_{\varphi}^{\star} - p_{\varphi}\right),
    \end{align}
\end{subequations}
where $\eta \in \mathbb{R}_{> 0}$ is the power droop gain, $\alpha \in \mathbb{R}_{\geq 0}$ is the voltage droop gain, $\omega_0$ is the nominal frequency, and the variables with superscript $^\star$ denote the associated setpoints. It is important to note that the power setpoints and feedback can be rotated by a tunable angle $\varphi$ to adapt to the network impedance characteristics \cite{juan2009adaptive}. Namely,
\begin{subequations}
\label{eq:rotated-power}
\begin{align}
\label{eq:rotated-power-setpoint}
    p_{\varphi}^{\star} + jq_{\varphi}^{\star} &\coloneqq e^{j(\pi/2 - \varphi)} (p^{\star} + jq^{\star}),\\
\label{eq:rotated-power-meas}
    p_{\varphi} + jq_{\varphi} &\coloneqq e^{j(\pi/2 - \varphi)} (p + jq),
\end{align}    
\end{subequations}
where the rotation angle ${\varphi}$ is chosen typically as the network impedance angle ${\varphi} = \arctan(\omega_0\ell/r)\in [0, \pi/2]$. For the dominantly inductive network case, ${\varphi} = \pi/2$ results in the standard decoupled p/f and q/v droop control.

\textit{2) Complex Droop Control:}
We augment the p/f droop control in \eqref{eq:classical-pf-droop} into a complex-power complex-frequency multivariable droop control (\textit{complex droop control} for short) as
\begin{equation}
\label{eq:complex-droop}
    \dot{\pha{\vartheta}} = \pha{\varpi}_0 + \eta e^{j\varphi} \left( \phaconj{\varsigma}^{\star} - \phaconj{\varsigma}\right),
\end{equation}
where $\pha{\varpi}_0 \coloneqq j{\omega}_0$ denotes the nominal complex frequency, $\eta \in \mathbb{R}_{> 0}$, $e^{j\varphi}$ denotes the rotation operator to adapt to the network impedance characteristics, and moreover, the complex power setpoint is given by $\phaconj{\varsigma}^{\star} \coloneqq \rho^{\star} - j\sigma^{\star} \coloneqq (p^{\star} - jq^{\star})/v^{\star 2}$. The second term in \eqref{eq:complex-droop} denotes the droop gain multiplied by the imbalance of complex power, which is to regulate the complex-frequency dynamics. From angular frequency to complex frequency, this complex droop control can be seen as an advancement of the elementary p/f droop control.

The complex droop control in \eqref{eq:complex-droop} is only able to regulate the rate of change of voltage (as well as angular frequency). To directly regulate the voltage amplitude, we introduce an amplitude regulation-enabled complex droop control \cite{he2023nonlinear} as
\begin{equation}
\label{eq:complex-droop-v}
    \dot{\pha{\vartheta}} = \pha{\varpi}_0 + \eta e^{j\varphi} \left( \phaconj{\varsigma}^{\star} - \phaconj{\varsigma}\right) + \eta \alpha \tfrac{v^{\star 2} - v^2}{v^{\star 2}},
\end{equation}
where $\alpha \in \mathbb{R}_{\geq 0}$ is the voltage regulation gain, $v^{\star}$ denotes the voltage setpoint, $v = e^{\Re{\pha{\vartheta}}}$ denotes the voltage amplitude, and $(v^{\star 2} - v^2)$ represents the voltage regulation deviation (another alternative is $(v^{\star} - v)$ \cite{colombino2019global}). We rewrite \eqref{eq:complex-droop-v} in real-valued polar coordinates as
\begin{subequations}
\label{eq:complex-droop-polar}
    \begin{align}
    {\dot v}/{v} &= \eta \left(\sigma_{\varphi}^{\star} - \sigma_{\varphi} \right) + \eta \alpha \left({v^{\star 2} - v^2}\right)/{v^{\star 2}},\\
    \dot \theta &= \omega_0  + \eta \left(\rho_{\varphi}^{\star} - \rho_{\varphi} \right),
    \end{align}
\end{subequations}
where the subscript $\varphi$ indicates the rotated power, i.e., 
\begin{subequations}
\label{eq:rotated-power-norm}
\begin{alignat}{2}
\label{eq:rotated-power-norm-setpoint}
    \sigma_{\varphi}^{\star} + j\rho_{\varphi}^{\star} &\coloneqq e^{j\varphi} \phaconj{\varsigma}^{\star} &&= (q_{\varphi}^{\star} + jp_{\varphi}^{\star})/{v^{\star 2}},\\
\label{eq:rotated-power-norm-meas}
    \sigma_{\varphi} + j\rho_{\varphi} &\coloneqq e^{j\varphi} \phaconj{\varsigma} &&= (q_{\varphi} + jp_{\varphi})/{v^{2}}.
\end{alignat}
\end{subequations}
The second equalities in \eqref{eq:rotated-power-norm-setpoint} and \eqref{eq:rotated-power-norm-meas} follow from the preceding related definitions.

The complex droop control in \eqref{eq:complex-droop-v} provides two control modes \cite{awal2021unified}, i.e., either with or without the voltage regulation term. The voltage-amplitude-following mode is given by $\alpha = 0$. The operation of this mode must rely on other voltage sources since the voltage amplitude is not controlled (though the frequency is controlled in a droop fashion). The voltage-forming mode, also termed grid-forming mode \cite{awal2021unified}, is given by $\alpha \neq 0$, where both the frequency and the voltage amplitude are controlled in a droop fashion.

\begin{remark}[Comparison between classical and complex droop controls]
The complex droop control in \eqref{eq:complex-droop-polar} differs from the classical droop control in \eqref{eq:classical-droop}. The comparison is summarized in Table~\ref{tab:comparison-droop-control}, where two relevant differences are identified:
\begin{itemize}
    \item The derivative of $v$ differs from that of $\ln v$, i.e., $\dot v/v$.
    \item The standard power feedback is different from the normalized power feedback.
\end{itemize}
The differences have a minor impact on the steady-state droop performance in the vicinity of $v \approx v^{\star} \approx 1.0$ pu, but they exert a significant impact on the \textit{transient stability property} due to the \textit{nonlinear coupling} from $p_{\varphi}$ and $q_{\varphi}$ to $v$ and $\theta$ in the large-signal regime, see \eqref{eq:complex-droop-polar} and \eqref{eq:rotated-power-norm}. These differences allow us to establish global asymptotic stability for complex droop control in this study, a result that remains unestablished in the general case for classical droop control \cite{schiffer2014conditions,schiffer2019global}. Furthermore, we have further established passivity and decentralized stability results for dVOC in \cite{he2023passivity}. The theoretically guaranteed passivity and stability properties motivate the upgrading of classical droop control with more advanced complex droop control in future grid-forming product development. Moreover, we note that since complex droop control allows an implementation in $v$ and $\theta$ coordinates as in \eqref{eq:complex-droop-polar}, similar to the implementation of classical droop control, established strategies concerning classical droop control, such as black-start, pre-synchronization, transition between islanded and grid-connected modes, and secondary control, can be readily applied and adapted to complex droop control. Additionally, complex droop control can be improved to provide richer dynamic ancillary services such as inertial response \cite{roger2023}.
\end{remark}

\begin{remark}[Equivalence to dVOC \cite{he2022complex}]
We recall the mathematical equivalence between complex droop control and dVOC \cite{he2022complex}, as illustrated in Fig.~\ref{fig:complex-droop}. Considering ${\pha{\dot v}}/{\pha{v}} = \pha{\dot \vartheta}$ following from the coordinate transformation $\pha{v} = e^{\pha{\vartheta}}$, we can rewrite the complex droop control in \eqref{eq:complex-droop-v}, from complex-angle coordinates to complex-voltage coordinates, as follows,
\begin{equation}
\label{eq:dvoc}
    \dot{\pha{v}} = \pha{\varpi}_0\pha{v} + \eta e^{j\varphi} \left( \phaconj{\varsigma}^{\star}\pha{v} - \pha{i} \right) + \eta \alpha \tfrac{v^{\star 2} - v^2}{v^{\star 2}}\pha{v},
\end{equation}
where $\pha{i} = \phaconj{\varsigma}\,\pha{v}$ follows from \eqref{eq:complex-power}.
The controller in \eqref{eq:dvoc} is the complex-variable version of the dVOC developed in \cite{gross2019effect}. In consideration of their equivalence, the terms ``dVOC" and ``complex droop control" can be used equivalently.
\end{remark}

% =================================================================================
\section{Modeling and Stability Problem Statement}
\label{sec:modeling}

% =======
% FIG
% =======
\begin{figure}
  \begin{center}
  \includegraphics{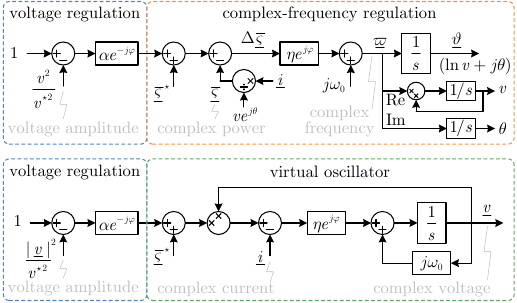}
  \caption{Complex droop control in complex-angle coordinates and equivalent dVOC in complex-voltage rectangular coordinates.}
  \label{fig:complex-droop}
  \end{center}
\end{figure}

% =======
% FIG
% =======
\begin{figure}
  \begin{center}
  \includegraphics{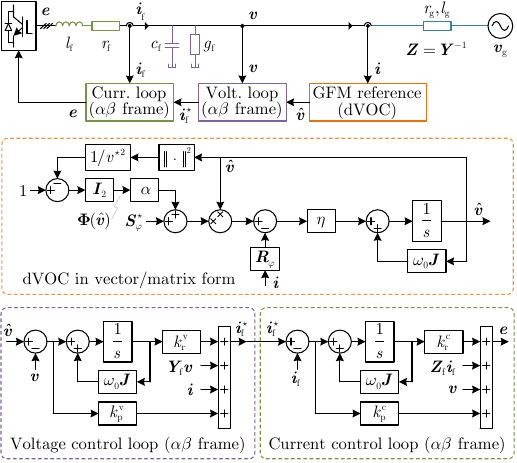}
  \caption{The block diagram of a typical grid-connected converter system, where dVOC (i.e., complex droop control) serves as a GFM control, and the inner voltage and current controls are implemented in the $\alpha\beta$ coordinate frame (also possible in the $dq$ coordinate frame).}
  \label{fig:system-detailed}
  \end{center}
\end{figure}

Consider a grid-connected converter system with complex droop control and cascaded voltage and current controls, as shown in vector/matrix form in Fig.~\ref{fig:system-detailed}. We first derive its full-order model, which includes the GFM control dynamics, the network electromagnetic dynamics, the filter dynamics, the voltage loop dynamics, and the current loop dynamics. For ease of modeling and analysis, we adopt the \textit{grid synchronous reference frame} oriented by the grid voltage vector angle $\theta_{\rm g}$ and consider that \textit{all AC signals are already transformed to the grid synchronous reference frame}. We define the main variables and parameters utilized throughout this paper in Table~\ref{tab:variables}. In the following, we use vectors $[a \enspace b]\trans$ and matrices $[\begin{smallmatrix}a & -b\\b & a \end{smallmatrix}]$ to equivalently represent the associated complex variables $a + jb$ for ease of derivation and understanding.

\begin{table}
\fontsize{7.8pt}{8pt}\selectfont
\centering
\caption{Definitions of Symbols, Variables, and Parameters}
\begin{tabular}{ll|ll}
\hline \hline
Symbol & Description & Symbol & Description\\
\hline
$v$ & Voltage amplitude & $\eta > 0$ & Power droop gain\\
$\theta$ & Voltage angle & $\alpha \geq 0$ & Voltage droop gain\\
$\pha v \coloneqq ve^{j \theta}$ & Complex voltage & $r_{\rm g}$ & Grid resistance\\
$\pha \vartheta \coloneqq \ln v + j \theta $ & Complex angle & $l_{\rm g}$ & Grid inductance\\
$\pha \varpi \coloneqq \dot{\pha \vartheta}$ & Complex frequency & $\pha z $& $r_{\rm g} + j \omega _{\rm g} l_{\rm g}$\\
$\varepsilon \coloneqq \dot v/v$ & \makecell[l]{Rate of change \\ of voltage} & $\pha y$ & $\pha z^{-1}$\\
$\omega \coloneqq \dot \theta$ & Angular frequency & $\mat Z$ & $[\begin{smallmatrix} r_{\rm g} & - \omega_{\rm g} l_{\rm g} \\ \omega_{\rm g} l_{\rm g} & r_{\rm g} \end{smallmatrix}]$ \\
$\pha{i}$ & Complex current & $\mat Y $ & $\mat Z^{-1}$ \\
$\pha{s} \coloneqq p + jq$ & Complex power & $r_{\rm f}$ & Filter resistance \\
$\pha{\varsigma} \coloneqq \pha{s}/v^2 $ & Normalized power & $l_{\rm f}$ & Filter inductance \\
$\rho \coloneqq p/v^2$ & Normalized $p$ & $g_{\rm f}$ & Filter conductance \\
$\sigma \coloneqq q/v^2$ & Normalized $q$ & $c_{\rm f}$ & Filter capacitance \\
$\varphi$ & Rotation angle & $\mat Y_{\rm f}$ & $ [\begin{smallmatrix} g_{\rm f} & - \omega_{\rm g} c_{\rm f} \\ \omega_{\rm g} c_{\rm f} & g_{\rm f} \end{smallmatrix}]$ \\
$p_{\varphi} + j q_{\varphi}$ & Rotated by $\varphi$ & $\mat Z_{\rm f}$ & $[\begin{smallmatrix} r_{\rm f} & - \omega_{\rm g} l_{\rm f} \\ \omega_{\rm g} l_{\rm f} & r_{\rm f} \end{smallmatrix}]$ \\
$\rho_{\varphi} + j \sigma_{\varphi}$ & Rotated by $\varphi$ & $\vect{i}_{\rm f}$ & Inductor current \\
$\vect{v}_{\rm g} \coloneqq [ v_{\rm g} \enspace 0 ]\trans$ & Grid voltage & $\vect{i}_{\rm f}^{\star}$ & Reference for $\vect{i}_{\rm f}$ \\
$\hat{\vect{v}} \coloneqq [ \hat v_{\rm d} \enspace \hat v_{\rm q} ]\trans $ & Reference voltage & $\vect e$ & Internal voltage \\
$\vect v \coloneqq [ v_{\rm d} \enspace v_{\rm q} ]\trans$ & Capacitor voltage & $\mat I_2$ & $[\begin{smallmatrix} 1 & 0 \\ 0 & 1 \end{smallmatrix}]$ \\
$\vect v_{\rm s}$ & Steady-state voltage & $\mat J$ & $[\begin{smallmatrix} 0 & -1 \\ 1 & 0 \end{smallmatrix}]$ \\
$\vect i$ & Output current & $\mat R_{\varphi}$ & $[\begin{smallmatrix} \cos \varphi & -\sin \varphi \\ \sin \varphi & \cos \varphi \end{smallmatrix}]$ \\
$\Re{}$ & Real part & $\mat S_{\varphi}^{\star} $ & $[\begin{smallmatrix} \sigma_{\varphi}^{\star} & -\rho_{\varphi}^{\star} \\ \rho_{\varphi}^{\star} & \sigma_{\varphi}^{\star} \end{smallmatrix}]$ \\
$\Im{}$ & Imaginary part & $\mat \Phi(\hat{\vect{v}})$ & $\frac{v^{\star 2} - \norm{\hat{\vect{v}}}^2}{v^{\star 2}} \mat I _2$ \\
$()\trans$ & Transpose & $\omega_{\rm g}, \theta_{\rm g}$ & Grid freq., angle \\
$()\hermconj$ & {Hermitian transpose} & $\omega_{\Delta} $ & $\omega_0 - \omega_{\rm g}$ \\
$()^\star$ & Setpoint & $\mat Y_{\varphi} $ & $\mat R_{\varphi} \mat Y$ \\
$()_{\rm s}$ & Steady state & $[\begin{smallmatrix} \kappa_{\rm r} & -\kappa_{\rm i} \\ \kappa_{\rm i} & \kappa_{\rm r}\end{smallmatrix}]$ & $\mat S_{\varphi}^{\star} - \mat Y_{\varphi} $\\
\hline \hline
\end{tabular}
\label{tab:variables}
\end{table}

\subsection{Full-Order System With Multi-Time-Scale Dynamics}

\textit{GFM Reference Model:}
The complex droop control in \eqref{eq:dvoc} serves as a GFM reference model to generate the reference voltage $\hat{\vect{v}}$ for the inner voltage loop, in vector form, as
\begin{equation}
\label{eq:dvoc-ref}
    \dot{\hat{\vect{v}}} = \omega_{\Delta} \mat J \hat{\vect{v}} + \eta  \mat S_{\varphi}^{\star} \hat{\vect{v}} - \eta \mat R_{\varphi} \vect i + \eta \alpha \mat \Phi(\hat{\vect{v}}) \hat{\vect{v}},
\end{equation}
where $\mat S_{\varphi}^{\star}$, as a matrix, encodes the rotated power setpoint $e^{j\varphi} \phaconj{\varsigma}^{\star}$, $\mat R_{\varphi}$ encodes $e^{j\varphi}$, $\mat \Phi(\hat{\vect{v}}) \coloneqq \frac{v^{\star 2} - \norm{\hat{\vect{v}}}^2}{v^{\star 2}} \mat I _2$, and $\norm{\cdot}$ denotes Euclidean distance. We note that $\omega_{\Delta} \coloneqq \omega_0 - \omega_{\rm g}$ instead of $\omega_0$ is used in the modeling from now on due to the adoption of the grid synchronous reference frame. One can refer to Table~\ref{tab:variables} for a complete description of variables/parameters.

\textit{Network Dynamics:}
The network electromagnetic dynamics are given as
\begin{equation}
\label{eq:network-dynamics}
    l_{\rm g} \dot{\vect{i}} + \mat Z \vect{i} = \vect v - \vect v_{\rm g},
\end{equation}
with the grid voltage $\vect{v}_{\rm g} \coloneqq [ v_{\rm g} \enspace 0 ]\trans$.

\textit{Filter Capacitor and Inductor Dynamics:}
The filter capacitor and inductor dynamics are given as
\begin{subequations}
\label{eq:filter-dynamics}
\begin{align}
\label{eq:capacitor-dynamics}
    c_{\rm f} \dot{\vect{v}} + \mat Y_{\rm f} \vect{v} + \vect{i} &= \vect{i}_{\rm f} ,\\
\label{eq:inductor-dynamics}
    l_{\rm f} \dot{\vect{i}}_{\rm f} + \mat Z_{\rm f} \vect{i}_{\rm f} + \vect v &= \vect e,
\end{align}
\end{subequations}
with the converter-side current $\vect{i}_{\rm f}$ and the internal voltage $\vect e$.

\textit{Voltage Controller:}
Consider a voltage reference tracking controller with feedforward compensation, which can be directly implemented in $\alpha\beta$ coordinates. When transformed into the grid reference frame, it is given as \cite{subotic2021lyapunov}
\begin{subequations}
\label{eq:voltage-controller}
    \begin{align}
        \dot {\vect \zeta} _{\rm v} &= \omega_{\Delta} \mat J \vect \zeta _{\rm v} + \vect v - \hat{\vect v},\\
        \vect i_{\rm f}^{\star} &= - k_{\rm p}^{\rm v} (\vect v - \hat{\vect v}) - k_{\rm r}^{\rm v} \vect \zeta _{\rm v} + \mat Y_{\rm f} \vect v + \vect i,
    \end{align}
\end{subequations}
where ${\vect \zeta} _{\rm v}$ stands for the resonant integrator state, $(\mat Y_{\rm f} \vect v + \vect i)$ denotes the feedforward compensation, and $k_{\rm p}^{\rm v} \in \mathbb{R}_{>0}$ and $k_{\rm r}^{\rm v} \in \mathbb{R}_{>0}$ are the voltage control gains.

\textit{Current Controller:}
Similarly, the current reference tracking controller in the grid reference frame is given as
\begin{subequations}
\label{eq:current-controller}
    \begin{align}
        \dot {\vect \zeta} _{\rm c} &= \omega_{\Delta} \mat J \vect \zeta _{\rm c} + \vect i_{\rm f} - \vect i_{\rm f}^{\star},\\
        \vect e &= - k_{\rm p}^{\rm c} (\vect i_{\rm f} - \vect i_{\rm f}^{\star}) - k_{\rm r}^{\rm c} \vect \zeta _{\rm c} + \mat Z_{\rm f} \vect i_{\rm f} + \vect v,
    \end{align}
\end{subequations}
where ${\vect \zeta} _{\rm c}$ stands for the resonant integrator state, $(\mat Z_{\rm f} \vect i_{\rm f} + \vect v)$ is the feedforward compensation, and $k_{\rm p}^{\rm c} \in \mathbb{R}_{>0}$ and $k_{\rm r}^{\rm c} \in \mathbb{R}_{>0}$ are the current control gains.

From \eqref{eq:dvoc-ref} to \eqref{eq:current-controller}, the full-order (twelfth-order) system is formulated as
\begin{subequations}
\label{eq:full-order-system-0}
    \begin{align}
        \tfrac{\rm d}{{\rm d}t} \underbrace{\hat{\vect{v}}}_{=:\vect x_1} &= \underbrace{\omega_{\Delta} \mat J \hat{\vect{v}} + \eta  \mat S_{\varphi}^{\star} \hat{\vect{v}} - \eta \mat R_{\varphi} \vect i + \eta \alpha \mat \Phi(\hat{\vect{v}}) \hat{\vect{v}}}_{=:\, \vect f_1(\vect x_1,\vect x_2)},\\
        \tfrac{\rm d}{{\rm d}t} \underbrace{\vect{i}}_{=:\vect x_2} &= \underbrace{l_{\rm g}^{-1} (-\mat Z \vect{i} + \vect v - \vect v_{\rm g})}_{=:\, \vect f_2(\vect x_1, \vect x_2, \vect x_3)},\\
        \tfrac{\rm d}{{\rm d}t} \underbrace{\begin{bmatrix} {\vect{v}} \\  {\vect \zeta} _{\rm v} \end{bmatrix}}_{=:\vect x_3} &= \underbrace{\begin{bmatrix} c_{\rm f}^{-1} (- \mat Y_{\rm f} \vect{v} - \vect{i} + \vect{i}_{\rm f}) \\ \omega_{\Delta} \mat J \vect \zeta _{\rm v} + \vect v - \hat{\vect v} \end{bmatrix}}_{=:\, \vect f_3(\vect x_1, \vect x_2, \vect x_3, \vect x_4)},\\
        \tfrac{\rm d}{{\rm d}t} \underbrace{\begin{bmatrix} {\vect{i}}_{\rm f} \\  {\vect \zeta} _{\rm c} \end{bmatrix}}_{=:\vect x_4} &= \underbrace{\begin{bmatrix} l_{\rm f}^{-1} (- k_{\rm p}^{\rm c} (\vect i_{\rm f} - \vect i_{\rm f}^{\star}) - k_{\rm r}^{\rm c} \vect \zeta _{\rm c}) \\ \omega_{\Delta} \mat J \vect \zeta _{\rm c} + \vect i_{\rm f} - \vect i_{\rm f}^{\star} \end{bmatrix}}_{=:\, \vect f_4(\vect x_1, \vect x_2, \vect x_3, \vect x_4)},\\
        \vect i_{\rm f}^{\star} &= - k_{\rm p}^{\rm v} (\vect v - \hat{\vect v}) - k_{\rm r}^{\rm v} \vect \zeta _{\rm v} + \mat Y_{\rm f} \vect v + \vect i,
    \end{align}
\end{subequations}
which is also graphically illustrated in Fig.~\ref{fig:full-reduced} from the slowest to the fastest dynamics in the form of a nested block diagram. We note that $\vect v_{\rm g}$ plays as a forced input, making the model different from that in the islanded case \cite[Sec. VI-C]{subotic2021lyapunov}. In particular, the system in \eqref{eq:full-order-system-0} becomes a linear one if the voltage-amplitude-following mode is used with $\alpha = 0$. This particular case is interesting in that the stability analysis becomes linear.

\subsection{Reduced-Order Systems}

\textit{1) Reduced Second-Order System With GFM Dynamics:} To focus solely on the GFM control dynamics, we can ignore the faster dynamics such that $\hat{\vect{v}} = \vect{v}$ (perfect voltage tracking) and $\tfrac{\rm d}{{\rm d}t} {\vect{i}} = \mathbbb{0}_2$ (static network), resulting in a reduced second-order model. In doing so, the network equation in \eqref{eq:network-dynamics} reduces to a static relationship as 
\begin{equation}
    \vect {i} = \mat Y \left(\vect {v} - \vect {v}_{\rm g}\right).
\end{equation}
The reduced second-order system is derived as
\begin{equation}
\label{eq:reduced-order-system}
    \tfrac{\rm d}{{\rm d}t} {\vect{v}} = [\omega_{\Delta} \mat J + \eta  (\mat S_{\varphi}^{\star} - \mat Y_{\varphi})] {\vect{v}} + \eta \mat Y_{\varphi} \vect v_{\rm g} + \eta \alpha \mat \Phi(\vect{v}) {\vect{v}},
\end{equation}
where $\mat Y_{\varphi} = \mat R_{\varphi} \mat Y$ denotes the rotated network admittance. We illustrate two special cases of \eqref{eq:reduced-order-system} in Appendix \ref{sec:approdix-a}, where their stability results can be readily obtained.

% =======
% FIG
% =======
\begin{figure}
  \begin{center}
  \includegraphics{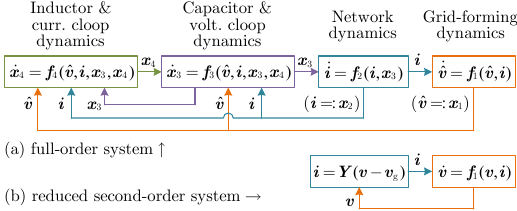}
  \caption{(a) The full-order model in \eqref{eq:full-order-system-0} in a multi-time-scale nested form. (b) The reduced second-order model in \eqref{eq:reduced-order-system} with GFM dynamics solely.}
  \label{fig:full-reduced}
  \end{center}
\end{figure}

\textit{2) Reduced Fourth-Order System:} We ignore the filter inductor dynamics, the current-loop dynamics, the filter capacitor dynamics, and the voltage-loop dynamics, while assuming $\hat{\vect{v}} = \vect{v}$, establishing a reduced fourth-order system, i.e.,
\begin{subequations}
\label{eq:fourth-order-system}
    \begin{align}
        \tfrac{\rm d}{{\rm d}t} \hat{\vect{v}} &= \omega_{\Delta} \mat J \hat{\vect{v}} + \eta  \mat S_{\varphi}^{\star} \hat{\vect{v}} - \eta \mat R_{\varphi} \vect i + \eta \alpha \mat \Phi(\hat{\vect{v}}) \hat{\vect{v}},\\
        \tfrac{\rm d}{{\rm d}t} \vect{i} &= l_{\rm g}^{-1} (-\mat Z \vect{i} + \hat{\vect{v}} - \vect v_{\rm g}).
    \end{align}
\end{subequations}

\textit{3) Reduced Eighth-Order System:} We ignore the filter inductor dynamics and the current-loop dynamics, while assuming $\vect i_{\rm f} = \vect i_{\rm f}^{\star}$, establishing a reduced eighth-order system, i.e.,
\begin{subequations}
\label{eq:eighth-order-system}
    \begin{align}
        \tfrac{\rm d}{{\rm d}t} \hat{\vect{v}} &= \omega_{\Delta} \mat J \hat{\vect{v}} + \eta  \mat S_{\varphi}^{\star} \hat{\vect{v}} - \eta \mat R_{\varphi} \vect i + \eta \alpha \mat \Phi(\hat{\vect{v}}) \hat{\vect{v}},\\
        \tfrac{\rm d}{{\rm d}t} \vect{i} &= l_{\rm g}^{-1} (-\mat Z \vect{i} + \vect v - \vect v_{\rm g}),\\
        \tfrac{\rm d}{{\rm d}t} \begin{bmatrix} {\vect{v}} \\  {\vect \zeta} _{\rm v} \end{bmatrix} &= \begin{bmatrix} c_{\rm f}^{-1} (- k_{\rm p}^{\rm v} (\vect v - \hat{\vect v}) - k_{\rm r}^{\rm v} \vect \zeta _{\rm v}) \\ \omega_{\Delta} \mat J \vect \zeta _{\rm v} + \vect v - \hat{\vect v} \end{bmatrix}.
    \end{align}
\end{subequations}

The full-order system and the different reduced-order systems have different levels of fidelity, and their applicability depends on the focus of the analysis scenario. We will provide the stability results for all these systems.

\begin{remark}[Comparison of modeling coordinates and analysis methods]
\label{rem:comparison}
Classical droop control and VSM are typically represented in polar coordinates as in \eqref{eq:classical-droop}. Correspondingly, their transient stability analysis is performed in polar coordinates, typically using the energy-function method \cite{shuai2019transient,choopani2020newmulti} or the phase-portrait method \cite{wu2019design,pan2020transient}. These methods focus mainly on the angle/frequency dynamics, while the voltage dynamics are not thoroughly considered since either a fixed voltage or a static voltage relationship is assumed \cite{shuai2019transient}. In contrast to polar coordinates, dVOC is developed initially in $\alpha\beta$ rectangular coordinates. The development is inspired by a consensus synchronization perspective (in rectangular coordinates), different from the heuristic design of droop control and VSM based on synchronous machines. The unique formulation of dVOC makes the stability analysis in rectangular coordinates feasible. Therefore, the rectangular-coordinate dVOC model will be employed in the subsequent transient analysis, which is also naturally compatible when further encompassing the inner-loop dynamics and the network dynamics. We will utilize the polar-coordinate model only in steady-state analysis (to reflect the benefit of power normalization). Furthermore, we will directly resort to the standard Lyapunov analysis method instead of the energy-function method to fully consider the dynamics of two rectangular coordinates. In contrast to the numerically solved phase-portrait representation, our analysis provides analytical results with quantitative stability conditions.
\end{remark}

\subsection{Stability Problem Statement}
\label{sec:statement-of-problem}

For the full-order and reduced-order systems given in \eqref{eq:full-order-system-0} to \eqref{eq:eighth-order-system}, we aim to explore the following stability problems.
\begin{enumerate}
    \item Under what conditions do equilibria exist?
    \item Under what conditions is an equilibrium globally stable?
    \item How will transient instability behave when unstable?
\end{enumerate}

These stability problems are closely relevant to the planning, design, analysis, and operation (particularly relevant to fault ride-through and transient operation) of GFM converters, drawing the attention of system planners, system operators, equipment manufacturers, and academic researchers.

\begin{remark}[Current limiting issue]
\label{rem:current-limiting}
The transient stability of GFM converters is also related to their current limiting strategies \cite{fan2022review}. When a grid fault occurs, the grid enters into an abnormal operation stage, during which converters should maintain GFM operation with auxiliary control/protection strategies. Ideally, our aim is for the converter under grid faults to resemble voltage sources and in the same model form as under normal conditions, e.g., arriving at an equivalent circuit model \cite{wu2024design}. This is motivated because, if so, we can apply the same approach in modeling, stability analysis, and system operation for both fault and non-fault conditions. This is indeed the case for conventional power systems. In another parallel work \cite{desai2023saturation}, we have designed a saturation-informed current limiting strategy that applies to various GFM controls, which is able to recover an equivalent normal form from the current-saturated operating condition under grid faults, achieving the above aim. Particularly, the strategy results in an equivalent circuit with a constant virtual impedance and an internal voltage source, and accordingly, the equivalent system conforms to the normal form. In this study, we focus on addressing the stability of the normal system form. Our results have been extended to current-saturated conditions \cite{desai2023saturation}.
\end{remark}

% Moreover, another typical current limiting strategy based on adaptive virtual impedance \cite{qoria2023variable} can partially fulfill the aim, where the varying virtual impedance for adapting in real-time to the severity of grid faults is unpredictable. 
% =================================================================================
\section{Stability of Second-Order GFM Dynamics}
\label{sec:stability analysis}

We first study the reduced second-order system in \eqref{eq:reduced-order-system} and then extend the results to the full-order system in \eqref{eq:full-order-system-0}. The main results are outlined in this and the next section to rapidly grasp the whole picture, while technical details are provided comprehensively in Appendices~\ref{sec:approdix-a} and~\ref{sec:appendix-b}. The technical details not only theoretically consolidate the stability analysis of dVOC but also inspire the stability analysis of other GFM controls. The advantage of our analysis method over previous energy-function or phase-portrait methods lies in the ability to fully account for the GFM dynamics and the higher-order dynamics and to gain analytical results.

\subsection{Existence of Equilibrium Points}
\label{sec:existence-of-ep}

The existence of equilibria is necessary for achieving stability in a steady state. We prove that equilibrium points of complex droop control always exist, see Proposition~\ref{prop:existence-of-ep} in Appendix~\ref{sec:approdix-a} for details. In contrast, there exist scenarios in which classical droop control does not admit a steady state, and a counterexample is illustrated in Example~\ref{exp:conunterexample} in Appendix~\ref{sec:approdix-a}. In respect thereof, classical droop control can suffer from transient instability due to the absence of equilibria\footnote{This fact has been partially documented in \cite{simpson2013synchronization,wu2019design}, where only the solution for phase angles is considered while voltages are assumed to be fixed. The solution for both phase angles and voltages is considered in this work.}. Unlike classical droop control, complex droop control does not suffer from this shortage. From the argument in the proof of Proposition~\ref{prop:existence-of-ep}, we reveal one reason for this is that the power feedback is normalized by the voltage amplitude square. Another reason is that a fixed setpoint for the normalized power is employed\footnote{A dVOC variant in \cite{awal2021transient} has been found to lose its equilibria in some cases, where $(p^{\star} - jq^{\star})/v^{2}$ serving as the setpoint for the normalized power is not fixed due to time-varying $v$.}, i.e., $\phaconj{\varsigma}^{\star} = (p^{\star} - jq^{\star})/v^{\star 2}$.

\subsection{Transient Stability of Complex Droop Control}
\label{sec:global-stability}

For a unique equilibrium $\vect v_{\rm s}$, a sufficient condition for global stability is identified by Theorem~\ref{thm:globally-stable} in Appendix~\ref{sec:approdix-a} as
\begin{equation}
\label{eq:global-stability-condition}
    \Re \bigl\{e^{j\varphi} \tfrac{p^{\star} - jq^{\star}}{v^{\star 2}} \bigr\} + \alpha < \tfrac{1}{2} \tfrac{\alpha}{v^{\star 2}} \norm {\vect v_{\rm s}}^2 + \Re \bigl\{e^{j\varphi} \pha{y} \bigr\}.
\end{equation}

\begin{table}
\footnotesize
\centering
\caption{Parametric Impact on the Stability of Complex Droop Control}
\begin{threeparttable}
\begin{tabular}{lll}
\hline \hline
Parameter change & Implication & Impact\\
\hline
$\abs{p^{\star} - jq^{\star}}$ $\uparrow$ & \makecell[l]{Enlarge the magnitude of \\ the power setpoint.} & Negative \\
\arrayrulecolor{black!10}\hline
$\lvert {\pha{y}} \rvert $ $\uparrow$ & \makecell[l]{Reduce the magnitude of \\the grid impedance\tnote{1}.} & Positive \\
\hline
$\alpha$ $\uparrow$ & \makecell[l]{Increase the gain of \\ the voltage amplitude regulation.} & \makecell[l]{Mostly \\ negative\tnote{2}} \\
\hline
$v_{\rm s}$ or $v_{\rm g}$ $\uparrow$ & \makecell[l]{Raise the voltage level \\ of the system.} & Positive \\
\hline
$\cos\bigl(\varphi - \angle \pha{z} \bigr)$ $\uparrow$ & \makecell[l]{Match the rotation angle with \\the grid impedance angle.} & Positive \\
\hline
$\cos\bigl(\varphi - \arctan\frac{q^{\star}}{p^{\star}}\bigr)$ $\uparrow$ & \makecell[l]{Match the rotation angle with \\ the power factor angle setpoint.} & Negative \\
\arrayrulecolor{black}\hline \hline
\end{tabular}
\label{tab:parametric-impact}
 \begin{tablenotes}
        \scriptsize
        \item[1] The grid impedance should incorporate the virtual impedance appearing for current limiting, implying that the virtual impedance compromises the transient stability.
        \item[2] $\alpha$ appears on both sides of \eqref{eq:global-stability-condition}, where the steady-state voltage value is also involved.
 \end{tablenotes}
\end{threeparttable}
\end{table}

The condition in \eqref{eq:global-stability-condition} relies on the equilibrium point information, which may be unknown. From \eqref{eq:global-stability-condition}, a tighter condition for global stability can be directly provided without using the equilibrium point information as
\begin{equation}
\label{eq:stability-condition-without-ep}
    \Re \bigl\{e^{j\varphi} \tfrac{p^{\star} - jq^{\star}}{v^{\star 2}} \bigr\} + \alpha < \Re \bigl\{e^{j\varphi} \pha{y} \bigr\}.
\end{equation}

\begin{remark}[Interpretation of the stability conditions]
\label{rem:physical-insights}
The stability conditions in \eqref{eq:global-stability-condition} and \eqref{eq:stability-condition-without-ep} provide quantitative insights into parametric impacts on stability, which are summarized in Table~\ref{tab:parametric-impact}. The first four pieces of insights are in line with well-known engineering experience. We provide an explanation for the last one. As seen from the model in \eqref{eq:reduced-order-system}, the real part of $e^{j \varphi} \phaconj{\varsigma}^{\star}$ regulates the $dq$-axis voltages with positive feedback. The match of the rotation angle ($\varphi$) with the power factor angle setpoint ($\arctan\tfrac{q^{\star}}{p^{\star}}$) leads to a large regulation gain, therefore not being conducive to stability. Moreover, we notice that the stability condition in \eqref{eq:global-stability-condition} for the grid-connected case resembles that in \cite[Condition 3]{he2023nonlinear} for the microgrid case.
\end{remark}

\begin{remark}[Grid strength of inductive-resistive networks]
Concerning a dominantly inductive grid impedance with $\pha y = -jb$ in p.u., the grid strength is typically represented by the short-circuit ratio (SCR), i.e., $\mathrm{SCR} \coloneqq b$ \cite{dong2019small}. We reveal from the stability condition in \eqref{eq:global-stability-condition} that for an inductive-resistive grid impedance with $\pha y = g-jb$ in p.u., the classical SCR definition may be extended as 
\begin{equation}
\label{eq:scr}
    \mathrm{SCR}_{\theta} \coloneqq \Re \bigl\{e^{j\varphi} \pha{y} \bigr\} = g \cos\theta + b\sin \theta,
\end{equation}
where $\theta$ is a control parameter reflecting the grid impedance angle \cite{juan2009adaptive}, see \eqref{eq:classical-droop} and \eqref{eq:complex-droop}. The inductive case, with $\theta = \pi/2$, is a special case of \eqref{eq:scr}. More generally, the results in \cite{dong2019small} and \cite{he2023passivity} provide further insights into the strength of a multi-node network and its relationship with small-signal and transient stability of grid-connected converter systems.
\end{remark}

\subsection{Transient Instability as Limit Cycle Oscillations}

We prove that the trajectories of complex droop control are always bounded, see Proposition~\ref{prop:bounded}. The trajectory boundedness implies that transient instability will manifest as limit cycles (periodic orbits) inside the bound. More precisely, if the trajectories of the system in \eqref{eq:reduced-order-system} cannot converge to equilibrium points, then they must converge to limit cycles, cf. Theorem~\ref{thm:transient-instability} in Appendix~\ref{sec:approdix-a}. Moreover, if there is only one equilibrium point and it is unstable, then all the initial states, except for the equilibrium point, converge to limit cycles, cf. Corollary~\ref{cor:limit-cycle}. The condition for the uniqueness of equilibrium points can be found in Proposition~\ref{prop:uniqueness-ep}, and the condition for an equilibrium point being unstable can be found in Proposition~\ref{prop:unstable-ep}. Namely, Propositions~\ref{prop:uniqueness-ep} and \ref{prop:unstable-ep} provide sufficient parametric conditions for transient instability.

\begin{remark}[Upper bound of the voltage profile]
\label{remark-upper-bound}
Under grid disturbances, there exists an upper bound for the converter voltage profile (even in the case of transient instability). The bound, denoted as $v_{\rm m}$, is identified in Proposition~\ref{prop:bounded} as
\begin{equation*}
    v_{\rm m} = \max \Bigl\{v_{\rm g},\, v^{\star} \Bigl(1 + \tfrac{\Re \bigl\{e^{j\varphi} \tfrac{p^{\star} - jq^{\star}}{v^{\star 2}} \bigr\} - \Re \{e^{j\varphi} \pha{y} \} + \lvert \pha{y} \rvert}{\alpha} \Bigr)^{1/2} \Bigr\}.
\end{equation*}
We remark that the existence of the upper bound can bring useful insights into the overvoltage protection of converters.
\end{remark}

% =================================================================================
\section{Stability of Full-Order System Dynamics}
\label{sec:extension}

This section extends our results from the reduced second-order system to the full-order system. The main focus is exploring the second problem stated in Section~\ref{sec:statement-of-problem}, i.e., under what conditions the equilibrium of the full-order system is asymptotically stable either globally or within a wide region of attraction. Concerning the other two problems, the full-order system shares the same equilibria as the reduced-order systems. Moreover, the transient instability of the full-order system, if dominated by the GFM dynamics, still manifests itself as limit cycle oscillations, as shown later in case studies.

To explore the transient stability of the full-order system, we resort to the nested singular perturbation approach developed in \cite{subotic2021lyapunov}. The approach aims to establish a stability guarantee for a nested interconnection of nonlinear dynamic systems ordered from slow to fast, where more than two time-scales can be considered, beyond classical singularly perturbed systems \cite[Theorem 11.3]{khalil2002nonlinear}. The approach has been applied to study converter-based islanded systems \cite{subotic2021lyapunov}. In this work, we apply it to the grid-connected case, where the grid voltage appears as a forced input, making the stability analysis largely different.

By applying the nested singular perturbation analysis (see Appendix~\ref{sec:appendix-b} for details), the following conditions are identified, which bound the controller set-points ($p^{\star}$, $q^{\star}$, $v^{\star}$), grid parameters ($\pha y$, $l_{\rm g}$, $r_{\rm g}$), and control parameters ($\alpha$, $\eta$, $\varphi$, ${k_{\rm p}^{\rm v}}$, ${k_{\rm r}^{\rm v}}$, ${k_{\rm p}^{\rm c}}$, ${k_{\rm r}^{\rm c}}$) to ensure the stability of the full-order system, i.e., 
\begin{subequations}
\label{eq:condition-for-full-order}
    \begin{align}
        \label{eq:condition-a}
        & \Re \bigl\{e^{j\varphi} \tfrac{p^{\star} - jq^{\star}}{v^{\star 2}} \bigr\} + \alpha + c_1 < \tfrac{1}{2} \tfrac{\alpha \norm {\vect v_{\rm s}}^2}{v^{\star 2}} + \Re \bigl\{e^{j\varphi} \pha{y} \bigr\},\\
        \label{eq:condition-b}
        &0 < \eta < \frac{c_1}{\frac{l_{\rm g}}{r_{\rm g}} \norm{\mat Y} (c_1 + c_{\epsilon})},\\
        \label{eq:condition-c}
        &0 < \frac{1 + {k_{\rm r}^{\rm v}}/{k_{\rm p}^{\rm v}}}{{k_{\rm r}^{\rm v}}/{c_{\rm f}} - 1} < \frac{4 c_2}{ \frac{1}{r_{\rm g}} \eta (c_{\epsilon}^2 + 4)},\\
        \label{eq:condition-d}
        &0 < \frac{1 + {k_{\rm r}^{\rm c}}/{k_{\rm p}^{\rm c}}}{{k_{\rm r}^{\rm c}}/{l_{\rm f}} - 1} < \frac{4 c_3}{\frac{\beta_{34}}{\Tilde{\beta}_{43}}(\Tilde{\beta}_{41}^2 + \Tilde{\beta}_{42}^2 + 4\Tilde{\beta}_{43}^2) + c_3\Tilde{\gamma}_4},
    \end{align}
\end{subequations}
where $c_1$, $c_2$, and $c_3$ quantify the stability margins concerning the unfavorable interaction appearing between adjacent time-scales of dynamics, and $c_{\epsilon}$ is a tunable quantity related to another tunable parameter $\epsilon$. The definition of all of these parameters/variables can be found in Appendix~\ref{sec:appendix-b}.

As a result of the grid voltage forced input in the grid-connected case, the stability conditions in \eqref{eq:condition-for-full-order} (particularly the first three) differ from those in \cite{subotic2021lyapunov} for islanded systems. However, they maintain similar physical interpretations. Broadly, the first inequality guarantees that the GFM dynamics alone are stable, the second one requires the GFM dynamics to be slow enough compared to the network dynamics, the third one implies that the voltage tracking should be sufficiently fast compared to the network dynamics, and finally, the last one necessitates that the current tracking should be faster than the voltage tracking (see \cite[Sec. VI-E]{subotic2021lyapunov} for further discussions). The above result can also be extended to multi-converter grid-connected systems, e.g., \eqref{eq:condition-a} has been extended as decentralized stability conditions using a passivity analysis \cite{he2023passivity}.

Concerning the reduced fourth-order and eighth-order systems in \eqref{eq:fourth-order-system} and \eqref{eq:eighth-order-system}, their stability conditions can be straightforwardly extracted from those of the full-order system. More precisely, the stability conditions of the former comprise \eqref{eq:condition-a} and \eqref{eq:condition-b}, while those of the latter comprise \eqref{eq:condition-a} to \eqref{eq:condition-c}.

\begin{remark}[Conservatism and practicality of the stability conditions]
The conservatism of the stability condition in \eqref{eq:condition-a} can be reflected by one of the \textit{instability} conditions given in Proposition \ref{prop:unstable-ep}, which reads as $\Re \bigl\{e^{j\varphi} \tfrac{p^{\star} - jq^{\star}}{v^{\star 2}} \bigr\} + \alpha > 2 \tfrac{\alpha}{v^{\star 2}} \norm {\vect v_{\rm s}}^2 + \Re \bigl\{e^{j\varphi} \pha{y} \bigr\}$. The only difference between the coefficients $1/2$ and $2$ reveals that a smaller steady-state voltage $\norm {\vect v_{\rm s}}$ implies less conservatism. The conservatism of \eqref{eq:condition-b} is illustrated in Fig.~\ref{fig:case-b1} in case studies, where the condition for $\eta$ is conservative around $1/2$ to $1/10$, affected by $\alpha$. Moreover, we have verified that \eqref{eq:condition-c} is easily satisfied. However, \eqref{eq:condition-d} is considerably conservative, requiring unrealistically high current control gains. Despite their conservatism, the quantitative conditions provide valuable parameter tuning guidelines by directing product research and development engineers on the relative order and direction of tuning of the cascaded loops. In doing so, the grid impedance information may take the rough value in the worst case, which may be given by system operators according to historical data. System operators can use the quantitative conditions to assess transient stability or dispatch system operation, where the equivalent grid impedance in the worst case can be screened out by considering different fault impedances and locations.
\end{remark}

\begin{remark}[Region of attraction]
Different from the global stability guarantee for the second-order system, we only obtain a non-global stability guarantee for the high-order systems\footnote{The stability condition derivations are not valid globally but in a certain region. This depends on the chosen Lyapunov function candidates. One may obtain global stability results with alternative Lyapunov function candidates.}. However, the region of attraction can still be guaranteed to be wide enough. Specifically, the region of attraction is given as $\{(\vect x_1,\vect x_2,\vect x_3,\vect x_4) \, \big\vert\, \norm{\vect x_1 - \vect v_{\rm s}} < r\}$, reflecting the largest sublevel set of the Lyapunov function $V_1$ on the largest neighbourhood $B_r \coloneqq \{\vect x_1 \, \big\vert\, \norm{\vect x_1 - \vect v_{\rm s}} < r\}$, where $r > 0$ is the positive real root of $\frac{(1+r/\norm{\vect v_{\rm s}})^3-1}{1/\norm{\vect v_{\rm s}}} = \epsilon$ (see Lemma~\ref{lem:neighborhood} and illustrative Example~\ref{exp:roa} in Appendix~\ref{sec:appendix-b} for further details). Hence, a larger $\epsilon$ implies a wider region of attraction. On the other hand, a large $\epsilon$ will downsize the control gains in \eqref{eq:condition-for-full-order} and accordingly slow down the whole dynamic response.
\end{remark}
% =================================================================================
\section{Simulation and Experimental Validations}
\label{sec:case-studies}

We present case studies to validate our previous theoretical results. First of all, we provide a comparative illustration of transient responses among the systems of different orders in \eqref{eq:full-order-system-0} to \eqref{eq:eighth-order-system}. Afterwards, we compare the analytical parametric stability range and associated simulation results. Thirdly, we validate that the transient instability of complex droop control manifests as limit cycle oscillations. Last, we validate that complex droop control achieves robust stability performance also in the case of non-ideal grid conditions. We remark that the full-order system encompasses all dynamics and thus serves as a high-fidelity model. The system parameters that are identically adopted in all simulation studies are given in the following: $S_N = 2$ MVA, $V_N = 690$ V, $f_N = 50$ Hz, $\omega_{\rm g} = 1.0$ pu, $v_{\rm g} = 1.0$ pu, $l_{\rm f} = c_{\rm f} = 0.05/\omega_0$ s/rad\footnote{The unit of s/rad is from the fact that the model as in \eqref{eq:full-order-system-0} is based on the per-unit system, where $l_{\rm g}$, $l_{\rm f}$, $c_{\rm f}$ [s/rad] contain $\omega_0$ (e.g., $\omega_0 = 100\pi$) as a divisor while $\eta$ [rad/s] contains $\omega_0$ as a factor.}, $r_{\rm f} = g_{\rm f} = 0.05/30$ pu, $k_{\rm p}^{\rm v} = 1$ pu, $k_{\rm r}^{\rm v} = 10$ pu/s, $k_{\rm p}^{\rm c} = 2$ pu, $k_{\rm r}^{\rm c} = 20$ pu/s. The other parameters are adjusted differently in each case study and therefore provided in their respective subsections. Beyond the simulations, the main results are also validated by experimental results based on real hardware converters.

\subsection{Case Study I: Transient Response for Different Orders}
\label{sec:case-study-a}

The following parameters are considered in this case study: $r_{\rm g} = 0.08$ pu, $l_{\rm g} = 0.2/\omega_0$ s/rad, $\varphi = \arctan\frac{l_{\rm g}}{r_{\rm g}}$, $p^{\star} = 0.5$ pu, $q^{\star} = 0.2$ pu, $v^{\star} = 1.0$ pu, $v_{\rm g,dip} = 0.5$ pu, and $\alpha = 1$ pu. The transient response of the four systems with different orders under a grid voltage dip is depicted in Fig.~\ref{fig:case-a}. In Fig.~\ref{fig:case-a}(a), a smaller value of $\eta = 0.02\omega_0$ is employed, which meets the stability condition in \eqref{eq:condition-b}. It can be observed that the full-order, eighth-order, and fourth-order systems exhibit similar transient responses. They, however, show deviations from the second-order system during the short period following the grid voltage dip. The deviation occurs due to the transient behaviors of circuit electromagnetic transients and/or inner-loop dynamics. From the result of decay rates, it is seen that the four different time-scales of dynamics in the full-order system roughly exhibit a time-scale separation from fast to slow. While the separation among the fast dynamics is not particularly significant, all of them operate at a faster rate compared to the dVOC GFM dynamics. Provided that circuit electromagnetic transients and inner-loop dynamics decay rapidly compared to the GFM dynamics, one may choose to use the second-order system, which only includes the dVOC GFM dynamics, for transient stability analysis.

In Fig.~\ref{fig:case-a}(b), a larger value of $\eta = 0.06\omega_0$ is utilized, which accelerates the dVOC GFM dynamics, violating the stability condition in \eqref{eq:condition-b}. Consequently, the GFM dynamics may interfere with the other dynamics. This interference is verified by the decay rate results, where the dVOC dynamics are even faster than the other dynamics. As a result, the electromagnetic dynamics exhibit poor transient characteristics, leading to noticeable oscillations, which take longer to decay. Therefore, in situations of severe interference, the second-order system may inadequately capture the transient response, necessitating the use of higher-order models. We note that similar results have been reported in prior studies pertaining to various GFM control strategies, such as in \cite{petrhighfidelity2018} concerning the interference between classical droop control and network dynamics. Our presented quantitative stability conditions in \eqref{eq:condition-for-full-order} provide valuable insights for mitigating such interference and offer guidance to engineers for the appropriate tuning of control gains in dVOC-controlled multi-time-scale converter systems. The following case study demonstrates that there exists an upper bound for $\eta$ to prevent the adverse interference between GFM dynamics and network dynamics.

% =======
% FIG
% =======
\begin{figure}
  \begin{center}
  \includegraphics{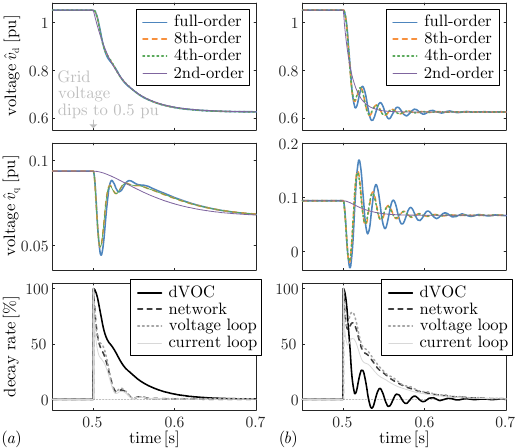}
  \caption{Transient response of systems of different orders. (a) $\eta = 0.02\omega_0$ rad/s, (b) $\eta = 0.06\omega_0$ rad/s. The decay rate of different time-scales of dynamics in the full-order system is represented by the error coordinates (normalized by their maximums): $\norm{\hat{\vect{v}} - \vect{v}_{\rm s}}/\norm{\hat{\vect{v}} - \vect{v}_{\rm s}}_{\max}$, $\norm{\vect{y}_2}/\norm{\vect{y}_2}_{\max}$, $\norm{\vect{y}_3}/\norm{\vect{y}_3}_{\max}$, and $\norm{\vect{y}_4}/\norm{\vect{y}_4}_{\max}$, see \eqref{eq:error-coodinates} in Appendix~\ref{sec:appendix-b} for the definition of the error coordinates $\vect{y}_i$, $i \in \{2,3,4\}$.}
  \label{fig:case-a}
  \end{center}
\end{figure}

\subsection{Case Study II: Parametric Stability Range}
\label{sec:case-study-b}

% =======
% FIG
% =======
\begin{figure}
  \begin{center}
  \includegraphics{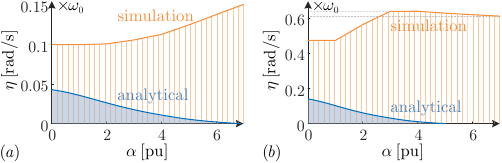}
  \caption{Comparison of the parametric stability ranges between the analytical result in \eqref{eq:condition-a} and \eqref{eq:condition-b} and the precise result by simulations traversing the parameter space. (a) $r_{\rm g} = 0.08$ pu and (b) $r_{\rm g} = 0.20$ pu.}
  \label{fig:case-b1}
  \end{center}
\end{figure}

% =======
% FIG
% =======
\begin{figure}
  \begin{center}
  \includegraphics{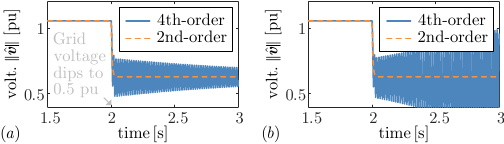}
  \caption{Concerning the critical stable case $(\alpha = 1, \eta = 0.100\omega_0)$ in Fig.~\ref{fig:case-b1}(a), the 4th-order system in \eqref{eq:fourth-order-system} is stable with $\eta = 0.099\omega_0$ in (a), whereas it is unstable with $\eta = 0.101\omega_0$ in (b).}
  \label{fig:case-b2}
  \end{center}
\end{figure}

We consider the reduced fourth-order system with stability conditions given in \eqref{eq:condition-a} and \eqref{eq:condition-b}, where two control gains of complex droop control, i.e., $\eta$ and $\alpha$, are constrained. We use the same parameters as in Section~\ref{sec:case-study-a} to exemplify the parametric stability range in terms of $\eta$ and $\alpha$.

The boundary of the analytical stability conditions in \eqref{eq:condition-a} and \eqref{eq:condition-b} is plotted by fixing a value of $\alpha$, solving the equilibrium, and finding the upper bound of $\eta$, as shown in Fig.~\ref{fig:case-b1}. The precise boundary is identified via simulations by traversing the parameter space. It can be seen that the analytical boundary is located strictly within the precise boundary as expected since the analytical results are sufficient but not necessary. By comparison between Fig.~\ref{fig:case-b1}(a) and (b), we remark that a larger $r_{\rm g}$ leads to a smaller time constant (faster decay) of the network dynamics, ${l_{\rm g}}/{r_{\rm g}}$, which allows a larger $\eta$, i.e., faster GFM dynamics. Moreover, it is verified by the precise boundary in Fig.~\ref{fig:case-b1}(b) that the impact of $\alpha$ on stability is not entirely negative, cf. the result in Table~\ref{tab:parametric-impact}, since the boundary is non-monotonically growing.

Consider a critical stable case, where $\alpha = 1$, $\eta = 0.100 \omega_0$ in Fig.~\ref{fig:case-b1}(a). It is observed from the simulation result in Fig.~\ref{fig:case-b2} that a larger $\eta$ immediately results in instability of the fourth-order system whereas a smaller $\eta$ can maintain stability. The second-order system, however, remains stable in both results. This is due to the fact that $\eta$ is not bounded in the stability condition \eqref{eq:global-stability-condition} for the second-order system. This suggests again that $\eta$ should not be overly large (the GFM dynamics should not be too fast) in the presence of network dynamics.

\subsection{Case Study III: Transient Instability as Limit Cycles}
\label{sec:case-study-c}

% =======
% FIG
% =======
\begin{figure}
  \begin{center}
  \includegraphics{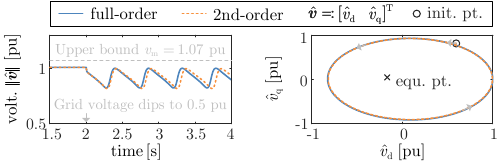}
  \caption{Transient instability of the complex-droop-controlled converter system behaves as a time-domain periodic oscillation or a phase-plane limit cycle.}
  \label{fig:case-c1}
  \end{center}
\end{figure}

% =======
% FIG
% =======
\begin{figure}
  \begin{center}
  \includegraphics{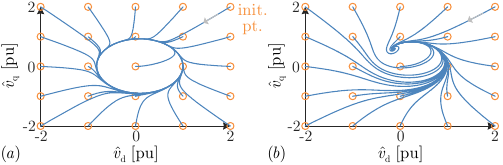}
  \caption{The complex-droop-controlled converter system (2nd-order) converges either to a limit cycle in (a), where $\alpha = 3$ pu, or to an equilibrium point in (b), where $\alpha = 1$ pu, from any non-equilibrium initial points.}
  \label{fig:case-c2}
  \end{center}
\end{figure}

The transient instability of complex droop control, if dominated by GFM dynamics, will manifest as periodic oscillations, since the second-order GFM dynamic trajectories are proved to be bounded. We consider the following parameters to illustrate the transient instability response: $r_{\rm g} = 0.8$ pu, $l_{\rm g} = 0.8/\omega_0$ s/rad, $\varphi = \arctan\frac{l_{\rm g}}{r_{\rm g}}$, $\eta = 0.08\omega_0$ rad/s, $\alpha = 3$ pu, $p^{\star} = 0.8$ pu, $q^{\star} = -0.2$ pu, $v^{\star} = 1.0$ pu, $v_{\rm g,dip} = 0.5$ pu. By Proposition~\ref{prop:uniqueness-ep}, it is identified that the system has a unique equilibrium. By Proposition~\ref{prop:unstable-ep}, it is identified that the equilibrium is unstable. Therefore, it follows from Corollary~\ref{cor:limit-cycle} that all the initial states of the GFM dynamics, except for the equilibrium point, converge to a limit cycle.

As seen in Fig.~\ref{fig:case-c1}, both the full- and second-order systems result in transient instability display limit cycle oscillations. The full-order system response is very close to the second-order one, implying that the instability is dominated by the GFM dynamics. Moreover, it is observed that the transient trajectory lies under the upper bound we identified, cf. Remark \ref{remark-upper-bound}, which suggests that the converter is free of overvoltage beyond the upper bound. To further observe the transient response of the GFM dynamics initiating from different initial points, we provide the phase portraits in Fig.~\ref{fig:case-c2}. In Fig.~\ref{fig:case-c2}(a), all the trajectories converge to the limit cycle. If we reduce $\alpha$ from $3$ to $1$, then a stable equilibrium point will be generated. Fig.~\ref{fig:case-c2}(b) shows that the GFM dynamics are stable in this case, as opposed to the transient instability in Fig.~\ref{fig:case-c2}(a), and the trajectories from all the initial points converge to the equilibrium point. For this stable case, however, we verified that the sufficient stability condition in \eqref{eq:global-stability-condition} is not satisfied. This can be justified because the system admits a larger parametric stability range than our analytical result, as illustrated in Fig.~\ref{fig:case-b1}.

\subsection{Case Study IV: Transient Stability Under Current Limiting}
\label{sec:case-study-d}

When a GFM converter suffers from current saturation under a grid fault, it is possible to represent it in an equivalent normal form based on an equivalent circuit \cite{desai2023saturation}, provided that a proper current limiting strategy is configured, as noted in Remark~\ref{rem:current-limiting}. This case study aims to validate the transient stability of complex droop control configured with the recent saturation-informed current limiting strategy under current limiting \cite{desai2023saturation}. We use a prototypical system as shown in Fig.~\ref{fig:case-d-system}, where the grid bus is represented in a stiff bus or more realistically a center-of-inertia model. The system parameters are given as $r_{\rm c} = 0.06$ pu, $l_{\rm c} = 0.30$ pu, $r_{\rm g} = 0.01$ pu, $l_{\rm g} = 0.05$ pu, $r_{\rm load} = 0.8$ pu, $l_{\rm load} = 1.2$ pu, $r_{\rm fault} = 0.01\ \Omega$, $\varphi = \arctan\frac{l_{\rm c}}{r_{\rm c}}$, $\eta = 0.03\omega_0$ rad/s, $\alpha = 3$ pu, $p^{\star} = 0.6$ pu, $q^{\star} = 0.2$ pu, $v^{\star} = 1.0$ pu. The generator capacity is the same as the converter, with the main parameters as $x_{\rm d}^{\prime} = 0.17$ pu, $r_{\rm s} = 0.038$ pu, $E = 1.1$ pu, $T_J = 4$ s, $D = 50$, $p_{\rm m} = 0.3$ pu.

% The performance of complex droop control is also crucial in practice, where the actual power grid is not an infinite bus. We use a prototypical system in Fig.~\ref{fig:case-d-system} to validate that complex droop control can achieve robust transient stability performance in a grid of synchronous generators. 
% =======
% FIG
% =======
\begin{figure}
  \begin{center}
  \includegraphics{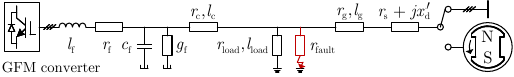}
  \caption{Simulation system for case study IV, where a GFM converter is connected to a stiff or center-of-inertia dynamic grid.}
  \label{fig:case-d-system}
  \end{center}
\end{figure}

% =======
% FIG
% =======
\begin{figure}
  \begin{center}
  \includegraphics{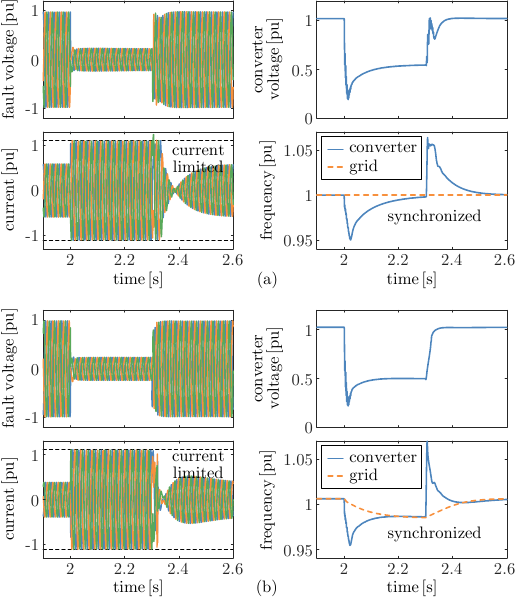}
  \caption{Under a grid fault, both transient stability and current limiting can be achieved using the saturation-informed current limiting strategy \cite{desai2023saturation}. (a) A stiff grid bus. (b) A center-of-inertia dynamic grid.}
  \label{fig:case-d}
  \end{center}
\end{figure}

A short-circuit fault occurs on the load bus at $2$~s, and the converter immediately reaches the current limit, as shown in Fig.~\ref{fig:case-d}. The saturation-informed current limiting strategy manages to limit the current at $1.1$ pu while maintaining dVOC's synchronizing and regulating capabilities. Consequently, the converter achieves synchronization with the grid during the fault. After the fault is cleared at $2.3$ s, the overcurrent exits, and the system maintains synchronization and transient stability, slowly settling down in the original steady state. In the case of a dynamic grid in Fig.~\ref{fig:case-d}(b), the generator frequency (i.e., the center-of-inertia dynamic grid frequency) drops down slowly since the grid supplies high active power to the fault resistance. Likewise, the converter achieves synchronization with the grid during the fault and after the fault clearance.

\subsection{Experimental Results}

% =======
% FIG
% =======
\begin{figure}
  \begin{center}
  \includegraphics{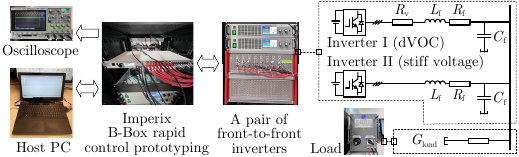}
  \caption{Experimental setup, where the nominal power capacity is $2$ kVA, the phase-to-phase RMS nominal voltage is $300$ V, and the filter parameters are as follows: $L_\mathrm{f} = 1.5$ mH, $R_\mathrm{f} = 1.0$ $\Omega$, and $C_\mathrm{f} = 3.5$ $\mu$F.}
  \label{fig:exp-setup}
  \end{center}
\end{figure}

% =======
% FIG
% =======
\begin{figure*}
  \begin{center}
  \includegraphics{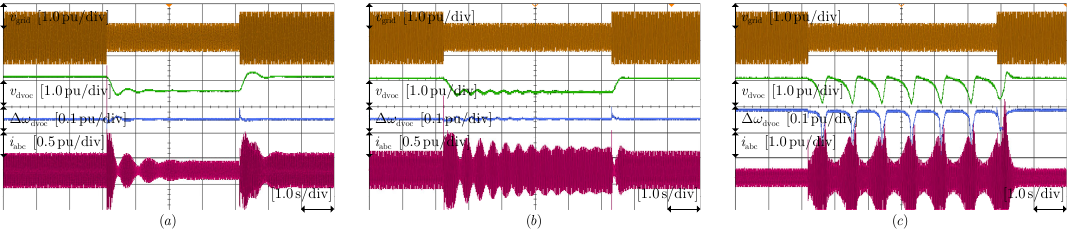}
  \caption{Experimental results indicate that the stability of complex droop control deteriorates with the increase of the voltage regulation gain $\alpha$, and the instability of complex droop control exhibits periodic oscillation. (a) $\alpha = 0$. (b) $\alpha = 1$. (c) $\alpha = 2$. The signal $v_\mathrm{grid}$ represents the grid voltage, $v_\mathrm{dvoc}$ and $\omega_\mathrm{dvoc}$ represent the control output of dVOC, and $i_\mathrm{abc}$ represent the output current of inverter I.}
  \label{fig:exp-results}
  \end{center}
\end{figure*}

The experimental setup depicted in Fig.~\ref{fig:exp-setup} is used to further validate the theoretical results. It involves the interconnection of two inverters in a front-to-front configuration, with one employing complex droop control and the other emulating a stiff grid. A virtual resistive line impedance $R_\mathrm{v}$ is implemented in inverter I (although it may also be implemented in inverter II) to enable the manipulation of the grid connection strength. The power generated by the inverters is then directed to a resistive load. Both inverters are controlled by an Imperix B-Box rapid control prototyping (RCP) system, which employs a dual-core ARM processor for the execution of control algorithms. The control code is compiled from a Simulink control model residing on the host computer. The RCP system supports real-time debugging and monitoring, allowing us to adjust control parameters online as well as capture signal waveforms through an oscilloscope. The PWM switching frequency is $32$ kHz, and the sampling and control frequency is $8$ kHz. It is noteworthy that the complex droop control in polar coordinates, instead of the dVOC in rectangular coordinates, is implemented in the control algorithm. In the case of choosing the dVOC, a higher control frequency may be required to ensure control accuracy, given the involvement of the integration of sinusoidal signals.

The experiments incorporate the following main parameters: $R_{\rm v} = 0.2$ pu, $G_\mathrm{load} = 0.56$ pu, $\varphi = \pi/2$, $p^{\star} = 0.3$ pu, $q^{\star} = 0.0$ pu, $v^{\star} = 1.0$ pu, $v_{\rm g,dip} = 0.5$ pu, and $\alpha \in \{0,1,2\}$ pu. The transient response of the system under the grid voltage dip to $v_{\rm g,dip}$ is shown in Fig.~\ref{fig:exp-results}. Notably, the system approaches marginal stability during the grid voltage dip, primarily due to the mismatch between the setting of $\varphi$ and the predominantly resistive network characteristic. Specifically, the system maintains transient stability with $\alpha = 0$ or $1$ while instability is observed with $\alpha = 2$. This observation is consistent with the theoretical results, indicating that the mismatch of $\varphi$ with the network impedance angle as well as the increase in $\alpha$ contributes to the deterioration of stability. Moreover, it is observed from Fig.~\ref{fig:exp-results}(c) that the transient instability manifests as periodic oscillation, consistent with the theoretical results. We notice that the inverter output current exhibits high spikes during transients or instability, while both inverters still remain operating, which is because the nominal capacity is reduced for safety compared to the actual capacity. In our upcoming studies, we will comprehensively assess the current limiting performance and transient stability of grid-forming controls under grid faults \cite{desai2023saturation}.

% =================================================================================
\section{Conclusion}

We have explored the transient stability of complex droop control in GFM converters connected to power grids. We prove that complex droop control achieves transient stability explicitly under the quantitative stability conditions we derived, demonstrating its superiority over classical droop control. The transient stability of complex droop control, in the multi-time-scale full-order dynamic system, also remains guaranteed within a wide region of attraction. The transient instability of complex droop control is shown to be bounded, manifesting itself as limit cycle oscillations. Our derived analytical stability/instability conditions and identified trajectory upper bound provide theoretical insights for practical parameter tuning, stability certificating, and stable operation of GFM converters. In another parallel work, we have explored a saturation-informed feedback control to address the current limiting issue in GFM converters, where the stability results obtained from the normal system form have been extended to guarantee transient stability under current saturation. Additionally, how a stable DC-bus voltage can be ensured for a GFM converter under the impact of AC-side GFM dynamics requires further investigation. Moreover, dynamic complex-frequency control of GFM converters for providing richer dynamic ancillary services is part of our future research.

% =================================================================================
\appendices

\section{Stability Results of the Second-Order System}
\label{sec:approdix-a}

We first provide two degenerated examples of the reduced second-order system and present their stability results.

\begin{example}[Voltage-amplitude-following mode]
\label{exp:voltage-following}
As mentioned in Section~\ref{sec:complex-droop-control}, if $\alpha = 0$ is employed, the GFM control will degenerate into a voltage-amplitude-following one. Then, the model in \eqref{eq:reduced-order-system} reduces to
\begin{equation}
\label{eq:model-grid-following}
    \dot{{\vect{v}}} = [\omega_{\Delta} \mat J + \eta  (\mat S_{\varphi}^{\star} - \mat Y_{\varphi})] {\vect{v}} + \eta \mat Y_{\varphi} \vect v_{\rm g}.
\end{equation}
The dynamics in \eqref{eq:model-grid-following} are \textit{linear}, and the stability analysis, in this case, is trivial, and provided in Proposition~\ref{prop:grid-follow-case}. We highlight that this voltage-amplitude-following mode proves to be globally stable and also naturally supports the grid frequency in a droop fashion, by contrast to classical grid-following control by phase-locked loops (PLLs) \cite{he2020transient}. Moreover, it is linear and thus tractable from the analysis point of view.
\end{example}

\begin{proposition}
\label{prop:grid-follow-case}
Consider the linear system in \eqref{eq:model-grid-following} (i.e., $\alpha = 0$). The system is globally asymptotically stable if and only if
\begin{equation}
\label{condi:grid-following}
    \Re \bigl\{e^{j\varphi} \tfrac{p^{\star} - jq^{\star}}{v^{\star 2}} \bigr\} < \Re \bigl\{e^{j\varphi} \pha{y} \bigr\}.
\end{equation}
\end{proposition}

\begin{proof}
The system matrix in \eqref{eq:model-grid-following}, i.e., $[\omega_{\Delta} \mat J + \eta  (\mat S_{\varphi}^{\star} - \mat Y_{\varphi})]$, is Hurwitz if and only if $\Re{e^{j \varphi}(\phaconj{\varsigma}^{\star} - \pha{y})} < 0$.
\end{proof}

\begin{example}[Off-grid mode]
\label{exp:off-grid}
Consider the other particular case where $\vect v_{\rm g} = \mathbbb{0}_2$ while $\alpha \neq 0$. This case means that the converter is off-grid/islanded to feed a constant impedance load. Accordingly, the model in \eqref{eq:reduced-order-system} reduces to, in the stationary coordinate frame,
\begin{equation}
\label{eq:model-offgrid}
    \dot{{\vect{v}}} = [\omega_0 \mat J + \eta  (\mat S_{\varphi}^{\star} - \mat Y_{\varphi}) + \eta \alpha \mat I_2] {\vect{v}} - \tfrac{\eta \alpha}{v^{\star 2}} \norm{\vect v}^2 {\vect{v}}.
\end{equation}
It is interesting to notice that the off-grid GFM dynamics in \eqref{eq:model-offgrid} follow the norm form of the Andronov-Hopf oscillator (also known as the Stuart-Landau oscillator) \cite{panteley2015stability}, i.e.,
\begin{equation}
\label{eq:adronov-Hopf-osc}
    \dot {\pha{z}} = \pha{\mu} \pha{z} - \pha{\xi} \abs{\pha{z}}^2 \pha{z},
\end{equation}
with gains $\pha{\xi}$, $\pha{\mu} \in \mathbb{C}$ and the state of the oscillator $\pha{z} \in \mathbb{C}$. The stability analysis of a single Andronov-Hopf oscillator is trivial and already available in \cite{panteley2015stability}. The stability result of \eqref{eq:model-offgrid} is directly provided in the following.
\end{example}

\begin{proposition}
\label{prop-off-grid-case}
The system in \eqref{eq:model-offgrid} (i.e., $\vect v_{\rm g} = \mathbbb{0}_2$) is globally asymptotically stable w.r.t the origin if and only if
\begin{equation}
\label{condi:offgrid}
    \Re \bigl\{e^{j\varphi} \tfrac{p^{\star} - jq^{\star}}{v^{\star 2}} \bigr\} + \alpha \leq \Re \bigl\{e^{j\varphi} \pha{y} \bigr\}.
\end{equation}
Otherwise, there exists a circular limit cycle, the amplitude square of which is $\norm{\vect v}^2 = \tfrac{v^{\star 2} }{\alpha} \Bigl(\Re \bigl\{e^{j\varphi} \tfrac{p^{\star} - jq^{\star}}{v^{\star 2}} - e^{j\varphi} \pha{y} \bigr\} + \alpha \Bigr)$.
\end{proposition}

\begin{proof}
    Consider the norm form of the Stuart-Landau oscillator in \eqref{eq:adronov-Hopf-osc}. It has been shown in \cite[Theorem 1]{panteley2015stability} that if $\Re \{\pha{\mu}\} \leq 0$, then the origin is globally asymptotically stable, while if $\Re \{\pha{\mu} \} > 0$, then the limit cycle $ \bigl\{ \lvert \pha{z} \rvert^2 = \Re \{\pha{\mu}\}/\Re \{\pha{\xi}\} \bigr\}$ is globally asymptotically stable for all initial states except for the origin. We apply this result to the system in \eqref{eq:model-offgrid}, which completes the proof.
\end{proof}

This limit-cycle convergence physically implies that the oscillator reaches a stable self-excited oscillatory state. The limit cycle is practically meaningful for the off-grid case since the AC voltage should remain in a stable oscillation.

From now on, we consider the general case for the second-order system in \eqref{eq:reduced-order-system}, where $\alpha \neq 0$ and $\vect v_{\rm g} \neq \mathbbb{0}_2$. In the general case, since $\vect v_{\rm g}$ plays as a forced input, the stability analysis becomes non-trivial. The stability analysis for this forced system is one of our theoretical contributions.

\subsection{Existence, Uniqueness, and Local Stability of Equilibria}

\begin{proposition}
\label{prop:existence-of-ep}
The GFM converter system in \eqref{eq:reduced-order-system} with complex droop control always has equilibrium points. In contrast, the GFM converter system with the classical droop control in \eqref{eq:classical-droop} or with VSM does not always have an equilibrium point.
\end{proposition}

\begin{proof}
To render the steady-state equation for complex droop control comparable to that for classical droop control, we consider both models in polar coordinates. For complex droop control in a steady state, we have, cf. \eqref{eq:complex-droop-polar}, that
\begin{subequations}
\label{eq:complex-droop-polar-ss}
    \begin{align}
    0 &= \eta \left(\sigma_{\varphi}^{\star} - \sigma_{\varphi\rm s} \right) + \eta \alpha \left({v^{\star 2} - v_{\rm s}^2}\right)/{v^{\star 2}}, \\
    0 &= \omega_{\Delta}   + \eta \left(\rho_{\varphi}^{\star} - \rho_{\varphi\rm s} \right),
    \end{align}
\end{subequations}
where the subscript ${\rm s}$ denotes steady-state variables (the same hereinafter). From \eqref{eq:rotated-power-norm-meas}, the rotated power is given as $\sigma_{\varphi\rm s} + j\rho_{\varphi\rm s} = e^{j\varphi} \phaconj{\varsigma}_{\rm s} = e^{j\varphi}(p_{\rm s} - jq_{\rm s})/{v_{\rm s}^{2}}$. Given that
\begin{subequations}
\label{eq:power-ss}
\begin{align}
    p_{\rm s} &= v_{\rm s}^2\lvert \pha{y} \rvert\cos{(\angle{\pha{z}})} - v_{\rm g} v_{\rm s} \lvert \pha{y} \rvert \cos \left(\delta_{\rm s} + \angle{\pha{z}} \right), \\
    q_{\rm s} &= v_{\rm s}^2 \lvert \pha{y} \rvert\sin{(\angle{\pha{z}})} - v_{\rm g} v_{\rm s} \lvert \pha{y} \rvert \sin \left(\delta_{\rm s} + \angle{\pha{z}} \right),
\end{align}
\end{subequations}
it follows that
\begin{subequations}
\label{eq:rotated-power-ss}
\begin{align}
    \sigma_{\varphi\rm s} &= \lvert \pha{y} \rvert\cos{\phi} - v_{\rm g} \lvert \pha{y} \rvert \cos \left(\delta_{\rm s} + \phi \right) /v_{\rm s}, \\
    \rho_{\varphi\rm s} &= - \lvert \pha{y} \rvert\sin{\phi} + v_{\rm g} \lvert \pha{y} \rvert \sin \left(\delta_{\rm s} + \phi \right) /v_{\rm s}.
\end{align}
\end{subequations}
where $\phi \coloneqq \angle{\pha{z}} - \varphi$ denotes the rotated grid-impedance angle, and $\delta_{\rm s}$ denotes the steady state of $\delta \coloneqq \theta - \theta_{\rm g}$.

Then, it follows by substituting \eqref{eq:rotated-power-ss} into \eqref{eq:complex-droop-polar-ss} that
\begin{subequations}
\label{eq:complex-droop-steady-state}
\begin{align}
    \sigma_{\varphi}^{\star} + \alpha - \alpha \tfrac{v_{\rm s}^2}{v^{\star 2}} &= \lvert \pha{y} \rvert\cos{\phi} - v_{\rm g} \lvert \pha{y} \rvert \cos \left(\delta_{\rm s} + \phi \right) /v_{\rm s}, \\
    \rho_{\varphi}^{\star} + \tfrac{\omega_{\Delta}}{\eta} &= - \lvert \pha{y} \rvert\sin{\phi} + v_{\rm g} \lvert \pha{y} \rvert \sin \left(\delta_{\rm s} + \phi \right) /v_{\rm s},
\end{align}
\end{subequations}
where there are two unknown variables $v_{\rm s}$ and $\delta_{\rm s}$. By eliminating $\delta_{\rm s}$, we obtain that
\begin{equation}
\label{eq:steady-voltage-sqr}
    \bigl[\sigma_{\varphi}^{\star} + \alpha - \alpha \tfrac{v_{\rm s}^2}{v^{\star 2}} - \lvert \pha{y} \rvert\cos{\phi} \bigr]^2 + \bigl[ \rho_{\varphi}^{\star} + \tfrac{\omega_{\Delta}}{\eta} + \lvert \pha{y} \rvert\sin{\phi} \bigr]^2 = \tfrac{v_{\rm g}^2 \lvert \pha{y} \rvert^2}{v_{\rm s}^2}.
\end{equation}

To look into the \textit{positive real root} of \eqref{eq:steady-voltage-sqr} for $v_{\rm s}^2$, we substitute $x \coloneqq v_{\rm s}^2 \in \mathbb{R}_{>0}$ and rewrite \eqref{eq:steady-voltage-sqr} as
\begin{multline}
\label{eq:cubic-equation}
    \underbrace{\bigl[\sigma_{\varphi}^{\star} + \alpha - \alpha \tfrac{x}{v^{\star 2}} - \lvert \pha{y} \rvert\cos{\phi} \bigr]^2 + \bigl[ \rho_{\varphi}^{\star} + \tfrac{\omega_{\Delta}}{\eta} + \lvert \pha{y} \rvert\sin{\phi} \bigr]^2}_{f(x),\, x \in \mathbb{R}_{>0}} \\
    = \underbrace{\tfrac{v_{\rm g}^2 \lvert \pha{y} \rvert^2}{x}}_{g(x),\, x \in \mathbb{R}_{>0}}.
\end{multline}
The left-hand side represents a parabola $f(x)$ opening upwards while the right-hand side represents a hyperbola branch $g(x)$ in the first quadrant. Since $f(x) \geq 0$, the equation \eqref{eq:cubic-equation} can only have \textit{positive} real roots. As geometrically shown in Fig.~\ref{fig:graphic-proof}, in the first quadrant, there must exist at least one intersection between $f(x)$ and $g(x)$. This is true because 
$\lim_{x \to 0} f(x) - g(x) = -\infty$ while $\lim_{x \to \infty} f(x) - g(x) = \infty$. Once a root for $x$ is solved from \eqref{eq:cubic-equation} (the root can be non-unique), we then choose $v_{\rm s} = \sqrt{x}$, since a positive voltage value is proper.

Given a solution of $v_{\rm s}$, we then use it to solve for $\delta_{\rm s}$ via \eqref{eq:complex-droop-steady-state}. Both the values of $\cos \left(\delta_{\rm s} + \phi \right)$ and $\sin \left(\delta_{\rm s} + \phi \right)$ are determined by applying the value of $v_{\rm s}$. Thus, there is only one unique solution for $\delta_{\rm s} + \phi$ within a periodic range $[\phi,\, 2\pi + \phi)$. Namely, one solution of $v_{\rm s}$ corresponds to one solution of $\delta_{\rm s}$ within $[0,\, 2\pi)$. The proof that complex droop control always has equilibrium points is now completed.

For classical droop control, VSM, or their variants with the same steady-state map, in a steady state, it holds that
\begin{subequations}
\label{eq:classical-droop-sss}
    \begin{align}
    0 &= \eta \left(q_{\varphi}^{\star} - q_{\varphi\rm s} \right) + \eta \alpha \left({v^{\star} - v_{\rm s}}\right), \\
    0 &= \omega_{\Delta} + \eta \left(p_{\varphi}^{\star} - p_{\varphi\rm s} \right).
    \end{align}
\end{subequations}
where the rotated power is referenced from \eqref{eq:rotated-power-ss} as
\begin{subequations}
\label{eq:rotated-power-sss}
\begin{align}
    q_{\varphi\rm s} &= \sigma_{\varphi\rm s}v_{\rm s}^2 = v_{\rm s}^2 \lvert \pha{y} \rvert\cos{\phi} - v_{\rm s}v_{\rm g} \lvert \pha{y} \rvert \cos \left(\delta_{\rm s} + \phi \right), \\
    p_{\varphi\rm s} &= \rho_{\varphi\rm s}v_{\rm s}^2 = -v_{\rm s}^2 \lvert \pha{y} \rvert\sin{\phi} + v_{\rm s}v_{\rm g} \lvert \pha{y} \rvert \sin \left(\delta_{\rm s} + \phi \right).
\end{align}
\end{subequations}
Similarly, it follows by substituting \eqref{eq:rotated-power-sss} into \eqref{eq:classical-droop-sss} and further eliminating $\delta_{\rm s}$ that
\begin{multline}
\label{eq:steady-state-classical-droop}
    \bigl[q_{\varphi}^{\star} + \alpha v^{\star} - \alpha v_{\rm s} - v_{\rm s}^2 \lvert \pha{y} \rvert\cos{\phi} \bigr]^2 + \\
    \bigl[ p_{\varphi}^{\star} + \tfrac{\omega_{\Delta}}{\eta} + v_{\rm s}^2 \lvert \pha{y} \rvert\sin{\phi} \bigr]^2 = v_{\rm g}^2 \lvert \pha{y} \rvert^2 v_{\rm s}^2,
\end{multline}
which is a quartic equation regarding $v_{\rm s}$. An arbitrary quartic equation does not necessarily have real roots, not to mention that we require \textit{positive} real roots for $v_{\rm s}$. It is not hard to find a counterexample, where \eqref{eq:steady-state-classical-droop} has no positive real roots, see Example~\ref{exp:conunterexample}. This completes the proof.
\end{proof}

% =======
% FIG
% =======
\begin{figure}
  \begin{center}
  \includegraphics{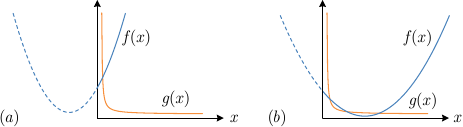}
  \caption{Equation \eqref{eq:cubic-equation} always has positive real roots, i.e., intersections between $f(x)$ and $g(x)$ in the first quadrant. (a) A unique positive real root. (b) At most three real positive roots.}
  \label{fig:graphic-proof}
  \end{center}
\end{figure}

\begin{example}
[Classical droop control lacks an equilibrium point]
\label{exp:conunterexample}
The steady-state equation in \eqref{eq:steady-state-classical-droop} has no real roots for $v_{\rm s}$ in the following case: $q_{\varphi}^{\star} = 0$ pu, $\alpha = 1$ pu, $v^{\star} = 1$ pu, $\varphi = \pi/2$, $p_{\varphi}^{\star} = 0$, $\omega_{\Delta} = 0$, $\eta = 0.08\omega_0$ rad/s, $v_{\rm g} = 0.1$ pu, $r_{\rm g} = 0.4$ pu, and $l_{\rm g} = 0.4$ pu, which can be verified numerically. This example implies that classical droop control may lose the existence of equilibrium points when the grid voltage is far away from $1.0$ pu. However, this violates the default assumption in classical droop control design, that is, it is only expected to operate around the nominal operating point. In comparison, complex droop control addresses this limitation by power normalization. We verify this example in simulation, where two independent grid-connected converter systems are considered, which are respectively controlled by classical droop control and complex droop control with the foregoing parameters. The simulation result of the two systems (both in full order) is shown in Fig.~\ref{fig:case-exp3}. The classical droop control cannot converge to a steady state under the severe grid voltage dip, while the complex droop control maintains transient stability and stabilizes at a new steady state. To theoretically guarantee the existence of equilibrium points for classical droop control, one can identify the analytical conditions from the quartic equation \eqref{eq:steady-state-classical-droop}. Alternatively, one can resort to using complex droop control, where the power feedback is scaled by the voltage square so that the existence of equilibria is guaranteed.

% =======
% FIG
% =======
\begin{figure}
  \begin{center}
  \includegraphics{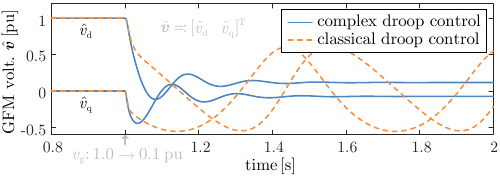}
  \caption{Classical droop control is unstable whereas complex droop control remains stable in Example~\ref{exp:conunterexample}.}
  \label{fig:case-exp3}
  \end{center}
\end{figure}

\end{example}

There exists at most three equilibria in the complex-droop-controlled converter system in \eqref{eq:reduced-order-system}, as graphically demonstrated in Fig.~\ref{fig:graphic-proof}. To yield global asymptotic stability, we require equilibria to be unique. The conditions for the uniqueness of equilibria are provided in the following.

\begin{proposition}
\label{prop:uniqueness-ep}
The equilibrium point of the system \eqref{eq:reduced-order-system} is unique if and only if the discriminant of the cubic equation from \eqref{eq:cubic-equation}, i.e., $a x^3 + b x^2 + c x + d = 0$, satisfies
\begin{equation}
\label{eq:delta}
    \Delta \coloneqq b^2c^2 - 4ac^3 - 4db^3 - 27a^2d^2 + 18abcd < 0,
\end{equation}
where the coefficients are given as 
\begin{equation*}
    \begin{aligned}
        a &\coloneqq {\alpha^2}/{v^{\star 4}},\\
        b &\coloneqq -2 {\alpha } \bigl(\sigma_{\varphi}^{\star} + \alpha - \lvert \pha{y} \rvert\cos{\phi} \bigr)/{v^{\star 2}},\\
        c &\coloneqq \bigl[\sigma_{\varphi}^{\star} + \alpha - \lvert \pha{y} \rvert\cos{\phi} \bigr]^2 + \bigl[ \rho_{\varphi}^{\star} + \tfrac{\omega_{\Delta}}{\eta} + \lvert \pha{y} \rvert\sin{\phi} \bigr]^2,\\
        d &\coloneqq -v_{\rm g}^2 \lvert \pha{y} \rvert^2.
    \end{aligned}
\end{equation*}
Moreover, a sufficient condition for the uniqueness of the equilibrium point of \eqref{eq:reduced-order-system} is $b > 0$, i.e.,
\begin{equation}
\label{condi:uniqueness}
    \Re \bigl\{e^{j\varphi} \tfrac{p^{\star} - jq^{\star}}{v^{\star 2}} \bigr\} + \alpha < \Re \bigl\{e^{j\varphi} \pha{y} \bigr\}.
\end{equation}

\end{proposition}

\begin{proof}
A cubic equation has a unique real root if and only if $\Delta < 0$ \cite{blinn2000polynomial}. Since we have shown in the proof of Proposition~\ref{prop:existence-of-ep} that the equation \eqref{eq:cubic-equation} has only \textit{positive} real roots, $\Delta < 0$ serves as a necessary and sufficient condition for the uniqueness of a \textit{positive} real root.

We then prove that $b > 0$ is a sufficient condition for $\Delta < 0$. Given the coefficients, it holds that $a > 0$, $c \geq 0$, $d < 0$, and $b^2 \leq 4ac$. It follows that $b^2c^2 - 4ac^3 \leq 0$ and $4b^2 \leq 16ac \leq 18 ac$. If $b > 0$, then $db < 0$ and further $- 4db^3 + 18abcd \leq 0$. Lastly, since $-27a^2d^2 < 0$, it follows that $\Delta = b^2c^2 - 4ac^3 - 4db^3 - 27a^2d^2 + 18abcd < 0$. Therefore, the real root of the cubic equation is unique. Namely, the equilibrium point of \eqref{eq:reduced-order-system} is unique.

Alternatively, we can directly observe from Fig.~\ref{fig:graphic-proof} that if the axis of symmetry of $f(x)$, $x_{\mathrm{sym}} \coloneqq \frac{v^{\star 2}}{\alpha} (\sigma_{\varphi}^{\star} + \alpha - \lvert \pha{y} \rvert\cos{\phi})$, satisfies $x_{\mathrm{sym}} < 0$, then there is only one intersection between $f(x)$ and $g(x)$ in the first quadrant, i.e., the cubic equation has a unique real root. The condition $x_{\mathrm{sym}} < 0$, equivalent to $b > 0$, serves as a sufficient condition for the uniqueness of the equilibrium point.
\end{proof}

\begin{proposition}
\label{prop:local-stability}
Consider an equilibrium point $\vect v_{\rm s}$ of the system \eqref{eq:reduced-order-system}. The equilibrium point is locally asymptotically stable if it holds that
\begin{equation}
\label{condi:locally-stable}
    \Re \bigl\{e^{j\varphi} \tfrac{p^{\star} - jq^{\star}}{v^{\star 2}} \bigr\} + \alpha < \tfrac{\alpha}{v^{\star 2}} \norm {\vect v_{\rm s}}^2 + \Re \bigl\{e^{j\varphi} \pha{y} \bigr\}.
\end{equation}
\end{proposition}

\begin{proof}
The Jacobian matrix of the system \eqref{eq:reduced-order-system} at the equilibrium point $\vect v_{\rm s} \coloneqq [v_{\rm ds}\enspace v_{\rm qs}]\trans$ is derived as
\begin{equation*}
    \mat{A} \coloneqq
    \begin{bmatrix}
        \eta \kappa_{\rm r} + \eta\alpha - \eta\alpha\tfrac{3v_{\rm ds}^2 + v_{\rm qs}^2}{v^{\star 2}}  & - \omega_{\Delta} - \eta \kappa_{\rm i} - \eta\alpha \tfrac{2v_{\rm ds}v_{\rm qs}}{v^{\star 2}} \\
        \omega_{\Delta} + \eta \kappa_{\rm i} - \eta\alpha \tfrac{2v_{\rm ds}v_{\rm qs}}{v^{\star 2}} & \eta \kappa_{\rm r} + \eta\alpha - \eta\alpha\tfrac{3v_{\rm qs}^2 + v_{\rm ds}^2}{v^{\star 2}}
    \end{bmatrix},
\end{equation*}
with auxiliary variables $\kappa_{\rm r} \coloneqq \Re \bigl\{e^{j\varphi} \tfrac{p^{\star} - jq^{\star}}{v^{\star 2}} \bigr\} - \Re \bigl\{e^{j\varphi} \pha{y} \bigr\}$ and $\kappa_{\rm i} \coloneqq \Im \bigl\{e^{j\varphi} \tfrac{p^{\star} - jq^{\star}}{v^{\star 2}} \bigr\} - \Im \bigl\{e^{j\varphi} \pha{y} \bigr\}$, see the definition in Table~\ref{tab:variables}. The characteristic equation is then given as
\begin{equation}
\begin{aligned}
    &{\lambda}^2 + 2\eta \bigl(\tfrac{2 \alpha \norm {\vect v_{\rm s}}^2}{v^{\star 2}} - \alpha - \kappa_{\rm r} \bigr) {\lambda} \, + \\
    & \qquad \bigl( \kappa_{\rm r} + \alpha - \tfrac{2 \alpha \norm {\vect v_{\rm s}}^2}{v^{\star 2}} \bigr)^2 - \bigl(\tfrac{\alpha \norm {\vect v_{\rm s}}^2}{v^{\star 2}} \bigr)^2 + \bigl(\tfrac{\omega_{\Delta}}{\eta} + \kappa_{\rm i} \bigr)^ 2 = 0.
\end{aligned}
\end{equation}
The equilibrium point is asymptotically stable if it holds that
\begin{equation}
\label{eq:local-stability}
    \left \{
    \begin{aligned}
        & \kappa_{\rm r} + \alpha < \tfrac{2 \alpha \norm {\vect v_{\rm s}}^2}{v^{\star 2}}, \\
        & \bigl(\kappa_{\rm r} + \alpha  - \tfrac{2\alpha \norm {\vect v_{\rm s}}^2}{v^{\star 2}} \bigr)^2 + \bigl(\tfrac{\omega_{\Delta}}{\eta} + \kappa_{\rm i}\bigr)^2 > \bigl(\tfrac{\alpha \norm {\vect v_{\rm s}}^2}{v^{\star 2}}\bigr)^2.
    \end{aligned} \right.
\end{equation}
A sufficient condition for \eqref{eq:local-stability} is $\kappa_{\rm r} + \alpha < \tfrac{\alpha}{v^{\star 2}} \norm {\vect v_{\rm s}}^2$, that is $\Re \bigl\{e^{j\varphi} \tfrac{p^{\star} - jq^{\star}}{v^{\star 2}} \bigr\} + \alpha < \tfrac{\alpha}{v^{\star 2}} \norm {\vect v_{\rm s}}^2 + \Re \bigl\{e^{j\varphi} \pha{y} \bigr\}$, by recalling the definition of $\kappa_{\rm r}$. This completes the proof.
\end{proof}

\begin{proposition}
\label{prop:unstable-ep}
Consider an equilibrium point $\vect v_{\rm s}$ of the system \eqref{eq:reduced-order-system}. The equilibrium point is unstable if either of the following two conditions is satisfied,
\begin{subequations}
    \begin{align}
        \kappa_{\rm r} + \alpha &> \tfrac{2 \alpha \norm {\vect v_{\rm s}}^2}{v^{\star 2}}\ \mathrm{or} \\
        \bigl(\kappa_{\rm r} + \alpha  - \tfrac{2\alpha \norm {\vect v_{\rm s}}^2}{v^{\star 2}} \bigr)^2 + \bigl(\tfrac{\omega_{\Delta}}{\eta} + \kappa_{\rm i}\bigr)^2 &< \bigl(\tfrac{\alpha \norm {\vect v_{\rm s}}^2}{v^{\star 2}}\bigr)^2.
    \end{align}
\end{subequations}
\end{proposition}
This proposition holds by directly referring to \eqref{eq:local-stability}.

\subsection{Global Asymptotic Stability of an Equilibrium Point}

\begin{theorem}
\label{thm:globally-stable}
Assume that the system \eqref{eq:reduced-order-system} has a unique equilibrium point $\vect v_{\rm s}$. The equilibrium point is globally asymptotically stable if it holds that
\begin{equation}
\label{condi:globally-stable}
    \Re \bigl\{e^{j\varphi} \tfrac{p^{\star} - jq^{\star}}{v^{\star 2}} \bigr\} + \alpha < \tfrac{1}{2} \tfrac{\alpha}{v^{\star 2}} \norm {\vect v_{\rm s}}^2 + \Re \bigl\{e^{j\varphi} \pha{y} \bigr\}.
\end{equation}
\end{theorem}

\begin{proof}
We define a positive-definite and radially unbounded Lyapunov candidate function w.r.t. $\vect v_{\rm s}$ as
\begin{equation}
\label{eq:lyapunov-func}
    V \coloneqq \tfrac{1}{2 \eta} \left(\vect v - \vect v_{\rm s}\right)\trans \left(\vect v - \vect v_{\rm s}\right),
\end{equation}
where $\vect v_{\rm s}$, as a unique equilibrium point, exclusively satisfies the steady-state equation of \eqref{eq:reduced-order-system} as
\begin{equation}
\label{eq:auxiliary-steady-state-equation}
    \mat Y_{\varphi} \vect v_{\rm g} = - [\tfrac{\omega_{\Delta} \mat J}{\eta} + (\mat S_{\varphi}^{\star} - \mat Y_{\varphi})] {\vect v_{\rm s}} - \alpha \mat \Phi(\vect v_{\rm s}) {\vect v_{\rm s}}.
\end{equation}
The Lyapunov function $V$ defined in \eqref{eq:lyapunov-func} physically represents the energy storage in the oscillator. The time derivative of $V$ along the system dynamics in \eqref{eq:reduced-order-system} is derived as follows,
\begin{alignat*}{2}
{\dot V} & = &&\ \tfrac{1}{\eta} \left(\vect v - \vect v_{\rm s}\right)\trans \dot {\vect v} \\
    & = &&\ \left(\vect v - \vect v_{\rm s}\right)\trans \bigl\{[\tfrac{\omega_{\Delta}}{\eta} \mat J + (\mat S_{\varphi}^{\star} - \mat Y_{\varphi})] {\vect{v}} + \mat Y_{\varphi} \vect v_{\rm g} + \alpha \mat \Phi(\vect{v}) {\vect{v}} \bigr\}\\
    & = &&\ \left(\vect v - \vect v_{\rm s}\right)\trans [\tfrac{\omega_{\Delta}}{\eta} \mat J + (\mat S_{\varphi}^{\star} - \mat Y_{\varphi}) + \alpha \mat I_2] \left(\vect{v} - \vect v_{\rm s} \right) \\
    & &&\ - \alpha \tfrac{1}{v^{\star 2}} [\left(\vect v - \vect v_{\rm s}\right)\trans (\norm{\vect v}^2 \vect v - \norm{\vect v_{\rm s}}^2 \vect v_{\rm s})] \\
    & = && \ (\kappa_{\rm r} + \alpha) \left(\vect v - \vect v_{\rm s}\right)\trans \left(\vect v - \vect v_{\rm s}\right)\\
    & &&\ - \alpha \tfrac{1}{v^{\star 2}} [\left(\vect v - \vect v_{\rm s}\right)\trans (\norm{\vect v}^2 \vect v - \norm{\vect v_{\rm s}}^2 \vect v_{\rm s})] \\
    & \leq && \ (\kappa_{\rm r} + \alpha - \alpha \tfrac{1}{v^{\star 2}} \tfrac{1}{2} \norm{\vect v_{\rm s}}^2) \left(\vect v - \vect v_{\rm s}\right)\trans \left(\vect v - \vect v_{\rm s}\right),
\end{alignat*}
where the last inequality holds due to Proposition~\ref{prop:inequality}. It follows that $\dot V$ is negative definite w.r.t. $\vect v_{\rm s}$, if it holds that $\kappa_{\rm r} + \alpha < \tfrac{1}{2}\tfrac{\alpha}{v^{\star 2}} \norm {\vect v_{\rm s}}^2$, which is equivalent to \eqref{condi:globally-stable} by recalling the definition of $\kappa_{\rm r}$ in Table~\ref{tab:variables}. This completes the proof.
\end{proof}

In this theorem, we require the uniqueness of the equilibrium point to ensure that $\vect v_{\rm s}$ in \eqref{eq:auxiliary-steady-state-equation} is unique and consistent with the equilibrium point used in the Lyapunov function. We refer to Proposition~\ref{prop:uniqueness-ep} for an identified necessary and sufficient condition as well as a sufficient one for the uniqueness of an equilibrium point.

\begin{proposition}
\label{prop:inequality}
$\forall \vect x, \vect y \in \mathbb{R}^2$, it holds that
\begin{equation*}
    (\vect x - \vect y) \trans (\norm{\vect x}^2 \vect x - \norm{\vect y}^2 \vect y) \geq \tfrac{1}{2} \norm{\vect y}^2 (\vect x - \vect y) \trans (\vect x - \vect y).
\end{equation*}
\end{proposition}
\begin{proof}
With the following argument,
\begin{align*}
    & (\vect x - \vect y) \trans \bigl(\norm{\vect x}^2 \vect x - \norm{\vect y}^2 \vect y \bigr) \\
    & = \norm{\vect x}^4 + \norm{\vect y}^4 - \vect x \trans \vect y \bigl(\norm{\vect x}^2 + \norm{\vect y}^2 \bigr) \\
    & \geq \tfrac{1}{2} \bigl(\norm{\vect x}^2 + \norm{\vect y}^2 \bigr)^2 - \vect x \trans \vect y \bigl (\norm{\vect x}^2 + \norm{\vect y}^2 \bigr) \\
    & = \tfrac{1}{2} \bigl(\norm{\vect x}^2 + \norm{\vect y}^2 \bigr) (\vect x - \vect y) \trans (\vect x - \vect y)\\
    & \geq \tfrac{1}{2} \norm{\vect y}^2 (\vect x - \vect y) \trans (\vect x - \vect y),
\end{align*}
the proof is completed.
\end{proof}

The condition in \eqref{condi:globally-stable} for global stability is a proper subset of the condition in \eqref{condi:locally-stable} for local stability. This is true because global stability suffices local stability. It is interesting to observe that the sufficient condition in \eqref{condi:uniqueness} for the uniqueness of the equilibrium point is consistent with the equilibrium point-independent sufficient condition in \eqref{eq:stability-condition-without-ep} for global stability. Moreover, the stability condition in \eqref{condi:globally-stable} for the general case resembles those for the degenerated results in \eqref{condi:grid-following} and \eqref{condi:offgrid}. This resemblance implies that the condition for the general can be seen as an extension of the degenerated. However, the former is sufficient whereas the latter ones are necessary and sufficient.

\subsection{Transient Instability Results}

\begin{proposition}
\label{prop:bounded}
The trajectories of the system in \eqref{eq:reduced-order-system} are bounded, and the ultimate voltage upper bound is given as 
\begin{equation}
\label{eq:ultimate-bound}
    v_{\rm m} \coloneqq \max \bigl\{v_{\rm g},\, v^{\star} \bigl(1 + \tfrac{\kappa_{\rm r} + \lvert \pha{y} \rvert}{\alpha} \bigr)^{1/2} \bigr\}.
\end{equation}
\end{proposition}

\begin{proof}
We define a Lyapunov candidate function as
\begin{equation}
    W \coloneqq \norm {\vect v}^2 = \vect v \trans \vect v,
\end{equation}
which represents the distance square of the voltage to the origin. The time derivative of $W$ along the system dynamics in \eqref{eq:reduced-order-system} is derived as follows,
\begin{equation}
\begin{aligned}
    \dot W &= 2 \vect v \trans \dot {\vect v} \\
    &= 2\eta \bigl(\kappa_{\rm r} + \alpha - \alpha \tfrac{\norm {\vect v}^2}{v^{\star 2}} \bigr) \norm {\vect v}^2 + 2\eta \vect v \trans \mat Y_{\varphi} \vect v_{\rm g}\\
    &\leq 2\eta \bigl(\kappa_{\rm r} + \alpha - \alpha \tfrac{\norm {\vect v}^2}{v^{\star 2}} \bigr) \norm {\vect v}^2 + 2\eta \lvert \pha{y} \rvert \norm {\vect v} v_{\rm g}.    
\end{aligned}
\end{equation}
When $\norm {\vect v} > v_{\rm g}$ and $\norm {\vect v} > v^{\star} \bigl(1 + \tfrac{\kappa_{\rm r} + \lvert \pha{y} \rvert}{\alpha} \bigr)^{1/2}$, it holds that
\begin{equation}
\begin{aligned}
    \dot W < 2\eta \bigl(\kappa_{\rm r} + \alpha + \lvert \pha{y} \rvert - \alpha \tfrac{\norm {\vect v}^2}{v^{\star 2}} \bigr) \norm {\vect v}^2 < 0.
\end{aligned}    
\end{equation}
With $v_{\rm m} = \max \bigl\{v_{\rm g},\, v^{\star} \bigl(1 + \tfrac{\kappa_{\rm r} + \lvert \pha{y} \rvert}{\alpha} \bigr)^{1/2} \bigr\}$ defined in \eqref{eq:ultimate-bound}, it directly follows that $\dot W < 0$ if $\norm {\vect v} > v_{\rm m}$. This suggests that $W$ will decrease monotonically outside the range $\{ \vect v\, \vert\, W(\vect v) \leq v_{\rm m}^2\}$. Therefore, each trajectory starting outside the range will decay until it enters the range, and each trajectory starting within the range will remain therein for all future time since $\dot W \leq 0$ on the boundary $W(\vect v) = v_{\rm m}^2$.
\end{proof}

\begin{theorem}
\label{thm:transient-instability}
If the system in \eqref{eq:reduced-order-system} cannot converge to an equilibrium point, then it must converge to limit cycles.
\end{theorem}
\begin{proof}
    Consider that the system is second-order and all the trajectories are bounded, cf. Proposition~\ref{prop:bounded}. The Poincar\'e-Bendixson Theorem \cite[Theorem 7.2.5]{miller1982ordinary} reveals that any bounded trajectory of second-order systems converges either to an equilibrium point or to a limit cycle.
\end{proof}

\begin{corollary}
\label{cor:limit-cycle}
For the system in \eqref{eq:reduced-order-system}, if there is only one equilibrium point and it is unstable, then all the initial states, except for the equilibrium point, converge to limit cycles.
\end{corollary}

This corollary is straightforward since a trajectory except for the equilibrium point itself cannot converge to an unstable equilibrium point. We note that this limit-cycle orbit is not necessarily a circle, see Fig.~\ref{fig:case-c1}, in contrast to the circular limit cycle in the off-grid case, see Example 2. This difference is attributed to the dominant effect of the grid voltage input on the oscillatory behavior.

\section{Stability Results of the Full-Order System}
\label{sec:appendix-b}

For ease of observation, the full-order multi-time-scale dynamic system in \eqref{eq:full-order-system-0} is given in the following again,
\begin{subequations}
\label{eq:full-order-system}
    \begin{align}
        \tfrac{\rm d}{{\rm d}t} \underbrace{\hat{\vect{v}}}_{=:\vect x_1} &= \underbrace{\omega_{\Delta} \mat J \hat{\vect{v}} + \eta  \mat S_{\varphi}^{\star} \hat{\vect{v}} - \eta \mat R_{\varphi} \vect i + \eta \alpha \mat \Phi(\hat{\vect{v}}) \hat{\vect{v}}}_{=:\, \vect f_1(\vect x_1,\vect x_2)},\\
        \tfrac{\rm d}{{\rm d}t} \underbrace{\vect{i}}_{=:\vect x_2} &= \underbrace{l_{\rm g}^{-1} (-\mat Z \vect{i} + \vect v - \vect v_{\rm g})}_{=:\, \vect f_2(\vect x_1, \vect x_2, \vect x_3)},\\
        \tfrac{\rm d}{{\rm d}t} \underbrace{\begin{bmatrix} {\vect{v}} \\  {\vect \zeta} _{\rm v} \end{bmatrix}}_{=:\vect x_3} &= \underbrace{\begin{bmatrix} c_{\rm f}^{-1} (- \mat Y_{\rm f} \vect{v} - \vect{i} + \vect{i}_{\rm f}) \\ \omega_{\Delta} \mat J \vect \zeta _{\rm v} + \vect v - \hat{\vect v} \end{bmatrix}}_{=:\, \vect f_3(\vect x_1, \vect x_2, \vect x_3, \vect x_4)},\\
        \tfrac{\rm d}{{\rm d}t} \underbrace{\begin{bmatrix} {\vect{i}}_{\rm f} \\  {\vect \zeta} _{\rm c} \end{bmatrix}}_{=:\vect x_4} &= \underbrace{\begin{bmatrix} l_{\rm f}^{-1} (- k_{\rm p}^{\rm c} (\vect i_{\rm f} - \vect i_{\rm f}^{\star}) - k_{\rm r}^{\rm c} \vect \zeta _{\rm c}) \\ \omega_{\Delta} \mat J \vect \zeta _{\rm c} + \vect i_{\rm f} - \vect i_{\rm f}^{\star} \end{bmatrix}}_{=:\, \vect f_4(\vect x_1, \vect x_2, \vect x_3, \vect x_4)},\\
        \vect i_{\rm f}^{\star} &= - k_{\rm p}^{\rm v} (\vect v - \hat{\vect v}) - k_{\rm r}^{\rm v} \vect \zeta _{\rm v} + \mat Y_{\rm f} \vect v + \vect i.
    \end{align}
\end{subequations}

\begin{theorem}
\label{thm:stability-of-full}
Consider the full-order system in \eqref{eq:full-order-system} with $\omega_{\Delta} = 0$ \footnote{The case where $\omega_{\Delta} \neq 0$ can be treated similarly (but with more lengthy algebraic manipulations), where $\mat P_3$ and $\mat P_4$ need to be adjusted accordingly.}, and let the conditions in \eqref{eq:condition-for-full-order} hold. Then, the equilibrium point of the system is asymptotically stable.
\end{theorem}

\begin{proof}
We perform a nested singular perturbation analysis on the full-order system given in \eqref{eq:full-order-system} through the subsequent procedures.

\textit{1) Steady-State Maps}: From the fast to slow dynamics in \eqref{eq:full-order-system}, we obtain the steady-state maps of the fourth, third, and second subsystems as
\begin{subequations}
    \begin{align*}
        \vect \phi_4 (\vect x_1, \vect x_2, \vect x_3) &= \begin{bmatrix}
            \vect i_{\rm f}^{\star} \\ \mathbbb{0}_2
        \end{bmatrix} = \begin{bmatrix}
            - k_{\rm p}^{\rm v} (\vect v - \hat{\vect v}) - k_{\rm r}^{\rm v} \vect \zeta _{\rm v} + \mat Y_{\rm f} \vect v + \vect i \\ \mathbbb{0}_2
        \end{bmatrix},\\
        \vect \phi_3 (\vect x_1, \vect x_2) &= \begin{bmatrix}
            \hat{\vect v} \\ \mathbbb{0}_2
        \end{bmatrix},\\
        \vect \phi_2 (\vect x_1) &= \mat Y (\hat{\vect {v}} - \vect {v}_{\rm g}),
    \end{align*}
\end{subequations}
by solving $\vect f_4(\vect x_1, \vect x_2, \vect x_3, \vect \phi_4) = 0$, $\vect f_3(\vect x_1, \vect x_2, \vect \phi_3, \vect \phi_4) = 0$, and $\vect f_2(\vect x_1, \vect \phi_2, \vect \phi_3) = 0$ successively. Furthermore, the equilibrium point of the full-order system can be defined as
\begin{equation*}
\begin{aligned}
    \begin{bmatrix}
        \vect x_{1\rm s} \\ 
        \vect x_{2\rm s} \\
        \vect x_{3\rm s} \\
        \vect x_{4\rm s}
    \end{bmatrix} &\coloneqq 
    \begin{bmatrix}
        \vect v_{\rm s} \\ 
        \vect \phi_2(\vect x_{1\rm s}) \\
        \vect \phi_3(\vect x_{1\rm s},\vect x_{2\rm s}) \\
        \vect \phi_4(\vect x_{1\rm s},\vect x_{2\rm s}, \vect x_{3\rm s})
    \end{bmatrix} = 
    \begin{bmatrix}
        \vect v_{\rm s} \\ 
        \mat Y ({\vect v_{\rm s}} - \vect {v}_{\rm g}) \\
        \begin{bmatrix}
            \vect v_{\rm s} \\ \mathbbb{0}_2
        \end{bmatrix}  \\
        \begin{bmatrix}
            \mat Y_{\rm f} \vect v_{\rm s} + \mat Y ({\vect v_{\rm s}} - \vect {v}_{\rm g}) \\ \mathbbb{0}_2
        \end{bmatrix}
    \end{bmatrix}.
\end{aligned}
\end{equation*}

\textit{2) Individual Reduced-Order Subsystem Representing Each Single Time-Scale of Dynamics:} With the steady-state maps, we define the reduced-order maps as
\begin{subequations}
\label{eq:vector-field}
    \begin{align}
        \vect f_1^{\rm s}(\vect x_1) &\coloneqq \vect f_1(\vect x_1, \vect \phi_2),\\
        \vect f_2^{\rm s}(\vect x_1,\vect x_2) &\coloneqq \vect f_2(\vect x_1, \vect x_2, \vect \phi_3),\\
        \vect f_3^{\rm s}(\vect x_1,\vect x_2,\vect x_3) &\coloneqq \vect f_3(\vect x_1, \vect x_2, \vect x_3, \vect \phi_4),
    \end{align}
\end{subequations}
which respectively represents the vector field of $i$th subsystem where the \textit{faster state} in $\vect f_i$ is in its steady states, i.e., $\vect x_{i+1} = \vect \phi_{i+1}$ for $i \in \{1,2,3\}$. For the purpose of analysis, we further define the error coordinates as
\begin{subequations}
\label{eq:error-coodinates}
    \begin{align}
        \vect y_2 &\coloneqq \vect x_2 - \vect \phi_2 = \vect i - \mat Y \left(\hat{\vect {v}} - \vect {v}_{\rm g}\right),\\
        \vect y_3 &\coloneqq \vect x_3 - \vect \phi_3 = \begin{bmatrix} {\vect{v}} - \hat{\vect v} \\  {\vect \zeta} _{\rm v} \end{bmatrix},\\
        \vect y_4 &\coloneqq \vect x_4 - \vect \phi_4 = \begin{bmatrix}
            \vect i_{\rm f} - \vect i_{\rm f}^{\star} \\ {\vect \zeta} _{\rm c}
        \end{bmatrix}.
    \end{align}
\end{subequations}

When additionally considering that all \textit{slower states} in $\vect f_2^{\rm s}, \vect f_3^{\rm s}, \vect f_4$ are constant, i.e., $\tfrac{\rm d}{{\rm d}t}{\vect x}_j = 0$ for all $j < i,\, i \in \{2,3,4\}$, then $\tfrac{\rm d}{{\rm d}t} {\vect \phi}_i(\vect x_1,\cdots,\vect x_{i-1}) = 0 \Rightarrow \tfrac{\rm d}{{\rm d}t} {\vect y}_i = \tfrac{\rm d}{{\rm d}t} {\vect x}_i$. Based on this, the individual reduced-order subsystem for each single time-scale of dynamics is defined as
\begin{subequations}
\label{eq:individual-reduced}
    \begin{align}
        \tfrac{\rm d}{{\rm d}t} \vect x_1 
        &= \vect f_1^{\rm s}(\vect x_1)\\
        & = [\omega_{\Delta} \mat J + \eta  (\mat S_{\varphi}^{\star} - \mat Y_{\varphi})] \vect x_1 + \eta \mat Y_{\varphi} \vect v_{\rm g} + \eta \alpha \mat \Phi(\vect x_1) \vect x_1,\notag \\ 
        \tfrac{\rm d}{{\rm d}t} \vect y_2 &= \tfrac{\rm d}{{\rm d}t} \vect x_2 = \vect f_2^{\rm s}(\vect x_1,\vect x_2) \\
        &= \vect f_2^{\rm s}(\vect x_1,\vect y_2 + \vect \phi_2) \notag\\
        & = - l_{\rm g}^{-1} \mat Z \vect{y}_2, \notag \\
        \tfrac{\rm d}{{\rm d}t} \vect y_3 &= \tfrac{\rm d}{{\rm d}t} \vect x_3 = \vect f_3^{\rm s}(\vect x_1,\vect x_2,\vect x_3) \\
        & = \vect f_3^{\rm s}(\vect x_1,\vect x_2,\vect y_3 + \vect \phi_3), \notag\\
        &= \begin{bmatrix} -c_{\rm f}^{-1} k_{\rm p}^{\rm v}\mat I_2 & -c_{\rm f}^{-1} k_{\rm r}^{\rm v} \mat I_2 \\ \mat I_2 & \omega_{\Delta} \mat J \end{bmatrix} \vect y_3, \notag \\
        \tfrac{\rm d}{{\rm d}t} \vect y_4 &= \tfrac{\rm d}{{\rm d}t} \vect x_4 = \vect f_4(\vect x_1,\vect x_2,\vect x_3,\vect x_4)\\
        &= \vect f_4(\vect x_1,\vect x_2,\vect x_3, \vect y_4 + \vect \phi_4), \notag\\
        &= \begin{bmatrix} -l_{\rm f}^{-1} k_{\rm p}^{\rm c}\mat I_2 & -l_{\rm f}^{-1} k_{\rm r}^{\rm c} \mat I_2 \\ \mat I_2 & \omega_{\Delta} \mat J \end{bmatrix} \vect y_4. \notag
    \end{align}
\end{subequations}
In these reduced-order subsystems, all faster states are in their steady-state maps and all slower states are constant.

\textit{3) Lyapunov Functions and Bounds of Their Decrease:} We consider the following Lyapunov candidate functions for the individual reduced-order subsystems in \eqref{eq:individual-reduced}, respectively,
\begin{subequations}
    \begin{align*}
        V_1 &\coloneqq \tfrac{1}{2 \eta} \left(\vect x_1 - \vect v_{\rm s}\right)\trans \left(\vect x_1 - \vect v_{\rm s}\right), \\
        V_2 & \coloneqq \tfrac{l_{\rm g}}{2 r_{\rm g}} \vect y_2 \trans \vect y_2, \\
        V_3 &\coloneqq \tfrac{1}{2} \vect y_3 \trans \mat P_3 \vect y_3, \\
        V_4 & \coloneqq \tfrac{1}{2} \vect y_4 \trans \mat P_4 \vect y_4,\\
        \mat P_3 &\coloneqq \begin{bmatrix} c_{\rm f}/k_{\rm p}^{\rm v} \mat I_2 & c_{\rm f}/k_{\rm r}^{\rm v} \mat I_2 \\ c_{\rm f}/k_{\rm r}^{\rm v} \mat I_2 & k_{\rm p}^{\rm v}/k_{\rm r}^{\rm v} \mat I_2 + k_{\rm r}^{\rm v}/k_{\rm p}^{\rm v} \mat I_2 \end{bmatrix} ,\\
        \mat P_4 &\coloneqq \begin{bmatrix} l_{\rm f}/k_{\rm p}^{\rm c} \mat I_2 & l_{\rm f}/k_{\rm r}^{\rm c} \mat I_2 \\ l_{\rm f}/k_{\rm r}^{\rm c} \mat I_2 & k_{\rm p}^{\rm c}/k_{\rm r}^{\rm c} \mat I_2 + k_{\rm r}^{\rm c}/k_{\rm p}^{\rm c} \mat I_2 \end{bmatrix}.
    \end{align*}
\end{subequations}
We remark that $\mat P_3$ and $\mat P_4$ are established for the case where $\omega_{\Delta} = 0$, which should be adjusted for $\omega_{\Delta} \neq 0$ in general.

For each individual reduced-order subsystem, we bound the decrease of the Lyapunov functions as follows,
\begin{subequations}
\label{eq:bound-for-individual}
    \begin{align}
        \tfrac{\partial V_1}{\partial \vect x_1} \vect f_1^{\rm s} &\leq -\alpha_1 \psi_1(\vect x_1)^2, \ \alpha_1 > 0, \\
        \tfrac{\partial V_2}{\partial \vect y_2} \vect f_2^{\rm s} &\leq -\alpha_2 \psi_2(\vect y_2)^2, \ \alpha_2 > 0, \\
        \tfrac{\partial V_3}{\partial \vect y_3} \vect f_3^{\rm s} &\leq -\alpha_3 \psi_3(\vect y_3)^2, \ \alpha_3 > 0, \\
        \tfrac{\partial V_4}{\partial \vect y_4} \vect f_4 &\leq -\alpha_4 \psi_4(\vect y_4)^2, \ \alpha_4 > 0,
    \end{align}
\end{subequations}
where $\alpha_1 = -\kappa_{\rm r} - \alpha + \alpha \tfrac{1}{v^{\star 2}} \tfrac{1}{2} \norm{\vect v_{\rm s}}^2 > 0$ and $\psi_1 = \norm{\vect x_1 - \vect v_{\rm s}}$ (cf. Theorem~\ref{thm:globally-stable}), $\alpha_2 = 1$ and $\psi_2 = \norm{\vect y_2}$, $\alpha_3 = 1 - c_{\rm f}/k_{\rm r}^{\rm v} > 0$ and $\psi_3 = \norm{\vect y_3}$, and $\alpha_4 = 1 - l_{\rm f}/k_{\rm r}^{\rm c} > 0$ and $\psi_4 = \norm{\vect y_4}$.

We bound the effect of neglecting faster dynamics in these reduced-order systems as follows,
\begin{subequations}
\label{eq:bound-for-fast}
    \begin{align}
        \tfrac{\partial V_1}{\partial \vect x_1} (\vect f_1 - \vect f_1^{\rm s}) & \leq \beta_{12} \psi_1\psi_2, \ \beta_{12} > 0, \\
        \tfrac{\partial V_2}{\partial \vect y_2} (\vect f_2 - \vect f_2^{\rm s}) & \leq \beta_{23} \psi_2\psi_3, \ \beta_{23} > 0,\\
        \tfrac{\partial V_3}{\partial \vect y_3} (\vect f_3 - \vect f_3^{\rm s}) & \leq \beta_{34} \psi_3\psi_4, \ \beta_{34} > 0,
    \end{align}
\end{subequations}
where $\beta_{12} = 1$, $\beta_{23} = 1/r_{\rm g}$, and $\beta_{34} = 1/k_{\rm p}^{\rm v} + 1/k_{\rm r}^{\rm v}$.

We further bound the effect of treating the slower states in the reduced-order systems as constants as follows,
\begin{subequations}
\label{eq:bound-for-slow}
    \begin{align}
        \label{eq:bound-for-slow-a}
        -\tfrac{\partial V_2}{\partial \vect y_2} \tfrac{\partial \vect \phi_2}{\partial \vect x_1} \vect f_1 & \leq b_{211}\psi_2\psi_1 + b_{221}\psi_2^2, \\
        \label{eq:bound-for-slow-b}
        -\tfrac{\partial V_3}{\partial \vect y_3} \tfrac{\partial \vect \phi_3}{\partial \vect x_1} \vect f_1 & \leq b_{311}\psi_3\psi_1 + b_{321}\psi_3\psi_2, \\
        \label{eq:bound-for-slow-c}
        -\tfrac{\partial V_3}{\partial \vect y_3} \tfrac{\partial \vect \phi_3}{\partial \vect x_2} \vect f_2 & \leq b_{312}\psi_3\psi_1 + b_{322} \psi_3\psi_2 + b_{332} \psi_3^2, \\
        \label{eq:bound-for-slow-d}
        -\tfrac{\partial V_4}{\partial \vect y_4} \tfrac{\partial \vect \phi_4}{\partial \vect x_1} \vect f_1 & \leq b_{411}\psi_4\psi_1 + b_{421}\psi_4\psi_2, \\
        \label{eq:bound-for-slow-e}
        -\tfrac{\partial V_4}{\partial \vect y_4} \tfrac{\partial \vect \phi_4}{\partial \vect x_2} \vect f_2 & \leq b_{412}\psi_4\psi_1 + b_{422} \psi_4\psi_2 + b_{432} \psi_4\psi_3, \\
        \label{eq:bound-for-slow-f}
        -\tfrac{\partial V_4}{\partial \vect y_4} \tfrac{\partial \vect \phi_4}{\partial \vect x_3} \vect f_3 & \leq b_{413}\psi_4\psi_1 + b_{423} \psi_4\psi_2 + b_{433} \psi_4\psi_3 + b_{443} \psi_4^2,
    \end{align}
\end{subequations}
where $b_{211} = c_{\epsilon} \eta \frac{l_{\rm g}}{r_{\rm g}} \norm{\mat Y} $, $b_{221} = \eta\frac{l_{\rm g}}{r_{\rm g}}\norm{\mat Y}$, $b_{311} = c_{\epsilon} c_{\rm v} \eta $, $b_{321} = c_{\rm v}\eta $, $b_{312} = b_{322} = b_{332} = 0$, $b_{411} = c_{\epsilon} c_{\rm c} \eta k_{\rm p}^{\rm v} $, $b_{421} = c_{\rm c} \eta k_{\rm p}^{\rm v} $, $b_{412} = 0$, $b_{422} = c_{\rm c}\norm{\mat Z}/l_{\rm g}$, $b_{432} = c_{\rm c}/l_{\rm g} $, $b_{413} = b_{423} = 0$, $b_{433} = c_{\rm c} \bigl(\norm{\mat Y_{\rm f} - k_{\rm p}^{\rm v}\mat I_2} (k_{\rm p}^{\rm v} + k_{\rm r}^{\rm v})/c_{\rm f} + k_{\rm r}^{\rm v} \bigr) $, and $b_{443} = c_{\rm c}\norm{\mat Y_{\rm f} - k_{\rm p}^{\rm v}\mat I_2}/c_{\rm f}$, and the employed
auxiliary variables are defined as $c_{\epsilon} \coloneqq \norm{\mat S_{\varphi}^{\star} - \mat Y_{\varphi} + \alpha \mat {I}_2} + \alpha \epsilon \frac{\norm{\vect v_{\rm s}}^2}{v^{\star 2}} $ with a tunable parameter $\epsilon > 0$, $c_{\rm v} \coloneqq c_{\rm f}/k_{\rm p}^{\rm v} + c_{\rm f}/k_{\rm r}^{\rm v}$, and $c_{\rm c} \coloneqq l_{\rm f}/k_{\rm p}^{\rm c} + l_{\rm f}/k_{\rm r}^{\rm c}$.

We can verify that all the bounds in \eqref{eq:bound-for-individual}, \eqref{eq:bound-for-fast}, and \eqref{eq:bound-for-slow} are satisfied globally, except for those in \eqref{eq:bound-for-slow-a}, \eqref{eq:bound-for-slow-b}, and \eqref{eq:bound-for-slow-d} being only valid for $\vect x_1$ (i.e., $\hat{\vect v}$) in the range of $\bigl\Vert \norm{\hat{\vect v}}^2 \hat{\vect v} - \norm{\vect v_{\rm s}}^2 \vect v_{\rm s} \bigr\Vert \leq \epsilon \norm{\vect v_{\rm s}}^2 \norm{\hat{\vect v} - \vect v_{\rm s}} $ instead of globally. To ensure that in this range there exists a neighborhood containing the equilibrium $\vect v_{\rm s}$, we require $\epsilon > 3$, as shown later in Lemma~\ref{lem:neighborhood}.

To establish a Lyapunov function for the full-order dynamics, we further define $\beta_{21}$ to $\beta_{43}$ and $\gamma_2$ to $\gamma_4$ as
\begin{subequations}
    \begin{alignat*}{4}
        &\underbrace{b_{211}}_{\beta_{21} > 0},\ && \underbrace{b_{221}}_{\gamma_{2}}, && &&\\
        &\underbrace{b_{311} + b_{312}}_{\beta_{31} > 0}, \ && \underbrace{b_{321} + b_{322}}_{\beta_{32} > 0}, \ &&\underbrace{b_{332}}_{\gamma_{3}}, && \\
        &\underbrace{b_{411} + b_{412} + b_{413}}_{\beta_{41} > 0}, \ && \underbrace{b_{421} + b_{422} + b_{423}}_{\beta_{42} > 0}, \ &&\underbrace{b_{432} + b_{433}}_{\beta_{43} > 0},\ && \underbrace{b_{443}}_{\gamma_{4}}.
    \end{alignat*}
\end{subequations}
We now consider the following composite Lyapunov candidate function for the full-order system in \eqref{eq:full-order-system},
\begin{equation}
    \nu \coloneqq \mu_1 V_1 + \mu_2 V_2 + \mu_3 V_3 + \mu_4 V_4,
\end{equation}
where $\mu_1 = 1$ and $\mu_i \coloneqq \prod_{j = 1}^{i-1}\frac{\beta_{j(j+1)}}{\beta_{(j+1)j}} > 0$ for $i\in\{2,3,4\}$. The time derivative of $\nu$ along the full-order dynamics is derived and bounded as \cite[Theorem 3]{subotic2021lyapunov}
\begin{equation}
    \dot \nu \leq -\begin{bmatrix}
        \psi_1 \\ \psi_2 \\ \psi_3 \\ \psi_4
    \end{bmatrix} \trans \mat M
    \begin{bmatrix}
        \psi_1 \\ \psi_2 \\ \psi_3 \\ \psi_4
    \end{bmatrix} ,
\end{equation}
where $\mat M \in \mathbb{R}^{4\times 4}$ is a symmetric matrix and it is recursively defined by its leading principal minors as follows,
\begin{equation}
\label{eq:definition-m}
    \mat M_i \coloneqq \begin{bmatrix}
        \mat M_{i-1} & - \vect \beta_i \mu_i \\ -\vect \beta_i\trans \mu_i & (\alpha_i - \gamma_i) \mu_i
    \end{bmatrix},\ i\in \{2,3,4\},
\end{equation}
with $\mat M_1 = \alpha_1$ and $\vect \beta_i \coloneqq [\frac{1}{2}\beta_{i1}\ \cdots\ \frac{1}{2}\beta_{i(i-2)}\ \beta_{i(i-1)}]\trans$.

\textit{4) Stability Conditions:} 
Using \cite[Proposition 1]{subotic2021lyapunov}, it is verified that the conditions in \eqref{eq:condition-for-full-order} guarantee $\mat M_i \succeq \mu_i c_i \mat I_i \succ 0$. The stability margins $c_1 \in (0,\, \alpha_1)$, and $c_2$ and $c_3$ are defined by $c_i \coloneqq \frac{1}{2}(\alpha_i - \gamma_i + \frac{\beta_{i(i-1)}}{\beta_{(i-1)i}} c_{i-1} - \sqrt{D_i})$ and $D_i \coloneqq (\alpha_i - \gamma_i + \frac{\beta_{i(i-1)}}{\beta_{(i-1)i}} c_{i-1})^2 + 4(\vect \beta_i\trans\vect \beta_i - (\alpha_i - \gamma_i) \frac{\beta_{i(i-1)}}{\beta_{(i-1)i}} c_{i-1})$ for $i\in \{2,3\}$ \cite[Condition 2]{subotic2021lyapunov}. The other parameters in \eqref{eq:condition-for-full-order} are given as $\Tilde{\beta}_{41} = \beta_{41}/c_{\rm c}$, $\Tilde{\beta}_{42} = \beta_{42}/c_{\rm c}$, $\Tilde{\beta}_{43} = \beta_{43}/c_{\rm c}$, $\Tilde{\gamma}_{4} = \gamma_{4}/c_{\rm c}$. Based on the foregoing results, we conclude by leveraging \cite[Theorem 3]{subotic2021lyapunov} that the equilibrium point is asymptotically stable.
\end{proof}

\begin{lemma}
\label{lem:neighborhood}
For a given $\vect x_{\rm s} \in \mathbb{R}^2 \backslash \{\mathbbb{0}_2$\}, there exists a non-empty neighborhood, $B_{r} \coloneqq \{\vect x \in \mathbb{R}^2 \, \big\vert\, \norm{\vect x - \vect x_{\rm s}} < r\}$, contained in the domain $D_{\epsilon} \coloneqq \bigr\{ \vect x \in \mathbb{R}^2 \, \big\vert\, \bigl\Vert \norm{\vect x}^2 {\vect x} - \norm{\vect x_{\rm s}}^2 \vect x_{\rm s} \bigr\Vert \leq \epsilon \norm{\vect x_{\rm s}}^2 \norm{\vect x - \vect x_{\rm s}} \bigl\}$, provided that $\epsilon > 3$. Moreover, given an $\epsilon > 3$, the boundary condition for the largest neighborhood $B_{r}$ is given by the positive real root of $\frac{(1+r/\norm{\vect x_{\rm s}})^3-1}{r/\norm{\vect x_{\rm s}}} = \epsilon$, where a larger $\epsilon$ allows a wider $B_{r}$.
\end{lemma}

\begin{proof}
For ease of exposition, we provide an equivalent proof in $\mathbb{C}$, where $\pha x_{\rm s} \in \mathbb{C} \backslash\{0\}$, $B_{r} = \{\pha x \in \mathbb{C} \, \big\vert\, \vert \pha x - \pha x_{\rm s}\vert < r\}$, and $D_{\epsilon} = \bigr\{ \pha x \in \mathbb{C} \, \big\vert\, \bigl\vert \vert \pha x \vert^2 {\pha x} - \vert \pha x_{\rm s} \vert ^2 \pha x_{\rm s} \bigr\vert \leq \epsilon \vert \pha x_{\rm s} \vert ^2 \vert \pha x - \pha x_{\rm s} \vert \bigl\}$.

We consider $\pha x_{\rm s} = 1$ without loss of generality because if $\pha x_{\rm s} \neq 1$, we can define
\begin{equation}
\label{eq:scaling}
    {\pha x}^{\prime} = \pha x/\pha x_{\rm s},\quad r^{\prime} \coloneqq r/\vert \pha x_{\rm s} \vert
\end{equation}
to obtain $\vert \pha x^{\prime} - 1\vert < r^{\prime}$ and  $\bigl\vert \vert {\pha x}^{\prime} \vert^2 {{\pha x}^{\prime}} - 1 \bigr\vert \leq \epsilon \vert {\pha x}^{\prime} - 1 \vert $, where the scaled $r^{\prime}$ serves as a different radius from $r$. For $\pha x_{\rm s} = 1$, we define the boundary of $B_r$ as
\begin{equation*}
    \pha x_{\theta} \coloneqq 1 + r \cos\theta + jr\sin\theta, \quad \theta \in [0,\, 2\pi).
\end{equation*}
We aim to find the lower bound of $\epsilon$ that guarantees $\pha x_{\theta},\, \forall \theta \in [0,\, 2\pi)$ to lie within $D_{\epsilon}$, i.e., $\bigl\vert \vert {\pha x}_{\theta} \vert^2 {{\pha x}_{\theta}} - 1 \bigr\vert \leq \epsilon \vert \pha x - \pha x_{\rm s} \vert = \epsilon r$.

Next, we show $\bigl\vert \vert {\pha x}_{\theta} \vert^2 {{\pha x}_{\theta}} - 1 \bigr\vert \leq (1+r)^3-1$ and the equality holds when $\theta = 0$. We have $\abs{\pha x_{\theta}}^2 = 1 + 2r\cos\theta + r^2$, which gives $\abs{\pha x_{\theta}}^2 \leq (1 + r)^2$ and $\big \vert\abs{\pha x_{\theta}}^2 - 1 \big \vert = r\abs{2\cos\theta + r} \leq r(r + 2)$. Then,
\begin{align*}
    \bigl\vert \vert {\pha x}_{\theta} \vert^2 {{\pha x}_{\theta}} - 1 \bigr\vert &= \bigl\vert \vert {\pha x}_{\theta} \vert^2 {{\pha x}_{\theta}} - \vert {\pha x}_{\theta} \vert^2 + \vert {\pha x}_{\theta} \vert^2 - 1 \bigr \vert \\
    &\leq \vert {\pha x}_{\theta} \vert^2 \vert {\pha x}_{\theta} - 1 \vert  + \bigl \vert \vert {\pha x}_{\theta} \vert^2 - 1 \bigr \vert \\
    &\leq (1+r)^2 r + r(r + 2) \\
    &= (1+r)^3-1.
\end{align*}
The equality holds when $\theta = 0$, where ${\pha x}_{\theta} = 1 + r$.

Given that $\bigl\vert \vert {\pha x}_{\theta} \vert^2 {{\pha x}_{\theta}} - 1 \bigr\vert \leq (1+r)^3-1$ holds, the condition $\frac{(1+r)^3-1}{r} \leq \epsilon$, i.e., $(1+r)^3-1 \leq \epsilon r$, suffices to guarantee $\bigl\vert \vert {\pha x}_{\theta} \vert^2 {{\pha x}_{\theta}} - 1 \bigr\vert \leq \epsilon r $, that is, to ensure that $\pha x_{\theta},\, \forall \theta \in [0,\, 2\pi)$ lie within $D_{\epsilon}$. We, then, use $f(r) \coloneqq \frac{(1+r)^3-1}{r}$ to identify the lower bound of $\epsilon$ such that there exists $r > 0$ satisfying $f(r) \leq \epsilon$. Since $f(r) $ is increasing in $\mathbb{R}_{>0}$ and $\lim_{r\to 0} f(r) = 3$, it follows that if $\epsilon > 3$, there will exist $r > 0$ such that $f(r) \leq 3$. As a consequence, $\epsilon > 3$ guarantees that $\pha x_{\theta},\, \forall \theta \in [0,\, 2\pi)$, the boundary of $B_r$, is contained in $D_{\epsilon}$.

Given an $\epsilon > 3$, the largest $B_{r}$ cooresponds to the positive real root of $f(r) = \frac{(1+r)^3-1}{r} = \epsilon$. In the general case, where $\pha x_{\rm s} \neq 1$, we replace $r$ with the scaled $r^{\prime} = r/\vert \pha x_{\rm s} \vert $ as shown in \eqref{eq:scaling}, giving the boundary condition as $\frac{(1+r/\vert \pha x_{\rm s} \vert)^3-1}{r/\vert \pha x_{\rm s} \vert} = \epsilon$. Since $\frac{(1+r/\vert \pha x_{\rm s} \vert)^3-1}{r/\vert \pha x_{\rm s} \vert}$ is strictly increasing for $r \in \mathbb{R}_{>0}$, a larger $\epsilon$ allows a wider neighbourhood $B_{r}$.
\end{proof}

\begin{example}[Region of attraction for the high-order systems]
\label{exp:roa}
Consider the parameters in Section~\ref{sec:case-study-a}. We have verified that with the control gains of $\alpha = 1$ pu, $\eta = 0.02\omega_0$ rad/s, $k_{\rm p}^{\rm v} = 1$ pu, and $k_{\rm r}^{\rm v} = 10$ pu/s, the stability conditions in \eqref{eq:condition-a} to \eqref{eq:condition-c} are easily satisfied for $3 < \epsilon \leq 11.4$. To satisfy \eqref{eq:condition-d}, unrealistically high gains for $k_{\rm p}^{\rm c}$ and $k_{\rm r}^{\rm c}$ are required. The current control loop with regular gains, however, remains stable in all case studies in Section~\ref{sec:case-studies}, where $k_{\rm p}^{\rm c} = 2$ pu, $k_{\rm r}^{\rm v} = 20$ pu/s are employed. This implies that \eqref{eq:condition-d} may be considerably conservative. To mitigate the conservatism, one can alternatively use $\mat M \succ 0$ ($\mat M$ defined in \eqref{eq:definition-m}) to evaluate the stability or choose alternative $\mat P_3$ and $\mat P_4$ in the Lyapunov functions. In this example, where we have identified the allowable range of $3 < \epsilon \leq 11.4$, the region of attraction for the high-order systems is illustrated in Fig.~\ref{fig:case-exp4}. It is seen that a large $\epsilon$ allows a wider region of attraction, where the largest region of attraction is identified with $\epsilon = 11.4$. The largest region of attraction is sufficient to cover the typical transient operating range of converters transitioning from a pre-fault point to a fault-on steady state.
\end{example}

\begin{remark}[Potential for wide application of the nested singular perturbation approach]
We finally remark that the nested singular perturbation analysis may apply to other GFM controls, which differ only in the GFM control dynamics while the other dynamics are largely the same. More generally, the application is widely possible for any nested model form of converter systems in the power electronics field, such as DC/DC and DC/AC converters and their hybrid interconnected microgrids or power grids.
\end{remark}

% =======
% FIG
% =======
\begin{figure}
  \begin{center}
  \includegraphics{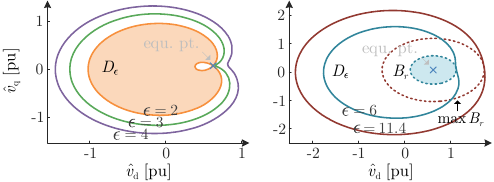}
  \caption{The neighborhood $B_r$, as an estimate of region of attraction within $D_{\epsilon}$, exists only if $\epsilon > 3$ is chosen. The largest region of attraction is identified with $\epsilon = 11.4$.}
  \label{fig:case-exp4}
  \end{center}
\end{figure}

% ====== REFERENCE SECTION
\bibliographystyle{IEEEtran}
\bibliography{IEEEabrv,Bibliography}

\begin{IEEEbiography}
[{\includegraphics[width=1in,height=1.25in,clip,keepaspectratio]{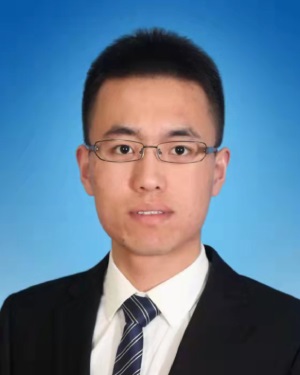}}]{Xiuqiang He (Member, IEEE)}
received his B.S. degree and Ph.D. degree in control science and engineering from Tsinghua University, China, in 2016 and 2021, respectively. Since 2021, he has been a Postdoctoral Researcher with the Automatic Control Laboratory, ETH Zürich, Switzerland. His current research interests include power system dynamics, stability, and control, involving multidisciplinary expertise in automatic control, power systems, power electronics, and renewable energy sources. Dr. He was the recipient of the Beijing Outstanding Graduates Award and the Outstanding Doctoral Dissertation Award from Tsinghua University.
\end{IEEEbiography}

\begin{IEEEbiography}
[{\includegraphics[width=1in,height=1.25in,clip,keepaspectratio]{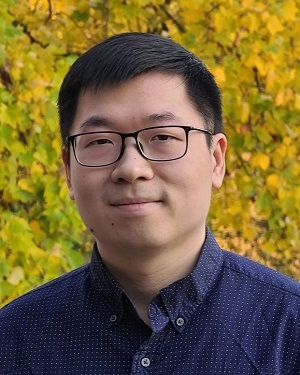}}]{Linbin Huang (Member, IEEE)}
received the B.Eng. and Ph.D. degrees from Zhejiang University, Hangzhou, China, in 2015 and 2020, respectively. Currently, he is a senior scientist with the Automatic Control Laboratory at ETH Zürich, 8092 Zürich, Switzerland. His research interests include power system stability, optimal control of power electronics, and data-driven control.
\end{IEEEbiography}

\begin{IEEEbiography}
[{\includegraphics[width=1in,height=1.25in,clip,keepaspectratio]{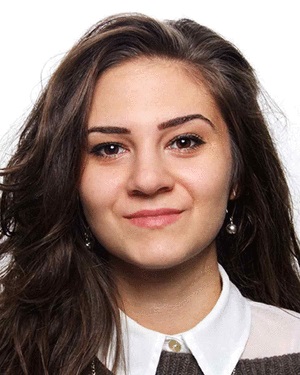}}]{Irina Subotić}
received the bachelor's degree in electrical engineering from the University of Belgrade, Belgrade, Serbia in 2016, and the master’s degree in robotics, systems and control from ETH Zürich, Zürich, Switzerland in 2018, and she has done her M.S. thesis in collaboration with the National Renewable Energy Laboratory (NREL), Golden, CO, USA. She received her Ph.D. degree from ETH Zürich, Zürich, Switzerland in 2024. She currently serves as a research scientist at ABB. Her research covers the control and optimization of converter-dominated power systems.
\end{IEEEbiography}

\begin{IEEEbiography}
[{\includegraphics[width=1in,height=1.25in,clip,keepaspectratio]{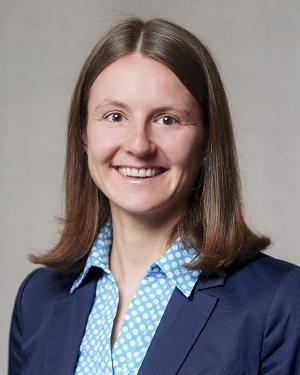}}]{Verena Häberle (Graduate Student Member, IEEE)}
received the B.Sc. and M.Sc. degrees in electrical engineering and information technology from ETH Zürich, Switzerland, in 2018 and 2020, respectively, where she is currently pursuing the Ph.D. degree with the Automatic Control Laboratory. For her outstanding academic achievements during her Master’s thesis with the Automatic Control Laboratory, ETH Zürich, under Prof. F. Dörfler, she was honored with ETH Medal and the SGA Award from the Swiss Society of Automatic Control. Her research focuses on the control design of dynamic virtual power plants in future power systems.
\end{IEEEbiography}

\begin{IEEEbiography}
[{\includegraphics[width=1in,height=1.25in,clip,keepaspectratio]{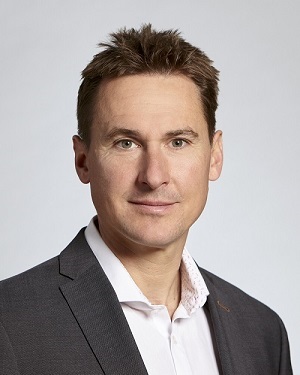}}]{Florian Dörfler (Senior Member, IEEE)}
is a Professor at the Automatic Control Laboratory at ETH Zürich. He received his Ph.D. degree in Mechanical Engineering from the University of California at Santa Barbara in 2013, and a Diplom degree in Engineering Cybernetics from the University of Stuttgart in 2008. From 2013 to 2014 he was an Assistant Professor at the University of California Los Angeles. He has been serving as the Associate Head of the ETH Zürich Department of Information Technology and Electrical Engineering from 2021 until 2022. His research interests are centered around automatic control, system theory, and optimization. His particular foci are on network systems, data-driven settings, and applications to power systems. He is a recipient of the distinguished young research awards by IFAC (Manfred Thoma Medal 2020) and EUCA (European Control Award 2020). His students were winners or finalists for Best Student Paper awards at the European Control Conference (2013, 2019), the American Control Conference (2016), the Conference on Decision and Control (2020), the PES General Meeting (2020), the PES PowerTech Conference (2017), the International Conference on Intelligent Transportation Systems (2021), and the IEEE CSS Swiss Chapter Young Author Best Journal Paper Award (2022). He is furthermore a recipient of the 2010 ACC Student Best Paper Award, the 2011 O. Hugo Schuck Best Paper Award, the 2012-2014 Automatica Best Paper Award, the 2016 IEEE Circuits and Systems Guillemin-Cauer Best Paper Award, the 2022 IEEE Transactions on Power Electronics Prize Paper Award, and the 2015 UCSB ME Best PhD award. He is currently serving on the council of the European Control Association and as a senior editor of Automatica.
\end{IEEEbiography}
\end{document}